\documentclass[12pt]{article}
\usepackage{indentfirst} 
\usepackage{amsfonts,amssymb,amsmath}
\usepackage{bm}          
\usepackage{cite}
\usepackage{url}
\usepackage{enumitem}
\usepackage[dvipsnames]{xcolor}
\usepackage{multicol}
\usepackage{MnSymbol}
\usepackage{hyperref}
\usepackage{tensor}
\usepackage{arcs}
\usepackage{tikz-cd}
\usepackage[titletoc,title]{appendix}

\usetikzlibrary{arrows}

\AtBeginDocument{}%

\newcommand{\corurl}{red}
\newcommand{\corcite}{ForestGreen}
\newcommand{\corlink}{blue}
\newcommand{\iprod}{\raisebox{8pt}{\scalebox{1}[-2.5]{$\neg\,$}}}

\newcommand{\dd}{\mathrm{d\!\raisebox{-0.05pt}{\scalebox{1}[1.021]{l}}}}
\newcommand{\w}{\wedge\hspace*{-6.77pt}{\raisebox{0.08pt}{\scalebox{.75}[.75]{$\wedge$}}}}
\newcommand{\ww}{\wedge\hspace*{-9.10pt}{\raisebox{0.08pt}{\scalebox{.75}[.75]{$\wedge$\,\,\,}}}}

\newcommand{\bd}{\bm{\mathrm{d}}}

\newcommand{\ee}{\bm{\mathrm{e}}}
\newcommand{\DD}{\bm{\mathrm{D}}}
\newcommand{\FF}{\bm{\mathrm{F}}}
\newcommand{\Aa}{\bm{\mathrm{A}}}

\hypersetup{linktocpage,colorlinks,urlcolor=\corurl,citecolor=\corcite,linkcolor=\corlink,
pdftitle={HK-GNH},
pdfauthor={J.F. Barbero, J. Margalef-Bentabol, A. Vicente-Cano, E.J.S. Villase\~nor}}

\numberwithin{equation}{section}  

\newtheorem{theorem}{Theorem}[section]

\def\QED{{\boldmath$\rule{0.5em}{0.5em}$}}                                
\def\markatright#1{\leavevmode\unskip\nobreak\quad\hspace*{\fill}{#1}}    
\def\qed{\markatright{\QED}}                                              

\newtheorem{lemma}[theorem]{Lemma}
\newenvironment{proof}[1][Proof]{\noindent\textbf{#1.} }{\qed\\}

\usepackage{titling}
\usepackage{authblk}
\title{Dynamical time and Ashtekar variables for the Husain-Kucha\v{r} model}
\author[1]{J. Fernando Barbero G.}
\author[2,3]{Juan Margalef-Bentabol\,}
\author[4]{Aitor Vicente-Cano\,}
\author[2,3]{Eduardo J.S. Villase\~nor\,}
\affil[1]{Instituto de Estructura de la Materia, IEM-CSIC. Serrano 123, 28006 Madrid, Spain}
\affil[2]{Departamento de Matem\'aticas, Universidad Carlos III de Madrid, Avda.\  de la Universidad 30, 28911 Legan\'es, Spain}
\affil[3]{Grupo de Teor\'{\i}as de Campos y F\'{\i}sica Estad\'{\i}stica. Instituto Gregorio Mill\'an (UC3M). Unidad Asociada al Instituto de Estructura de la Materia, CSIC}
\affil[4]{Departament de F\'{i}sica Qu\`antica i Astrof\'{i}sica, Institut de Ci\`encies del Cosmos,
Universitat de Barcelona, Mart\'{i} i Franqu\`es 1, 08028 Barcelona, Spain}
\date{}                     
\begin{document}
	
\maketitle
\date{September 30, 2024}


\begin{abstract}
The relation between the Husain-Kucha\v{r} model and some extensions thereof that incorporate a dynamical time variable is explored. To this end, we rely on the geometric approach to the Hamiltonian dynamics of singular (constrained) systems provided by the Gotay-Nester-Hinds method (GNH). We also use recent results regarding the derivation of the Ashtekar formulation for Euclidean general relativity without using the time gauge. The insights provided by the geometric approach followed in the paper shed light on several issues related to the dynamics of general relativity and the role of time.
\end{abstract}

\tableofcontents

\medskip
\noindent
{\bf Key Words:}
Husain-Kucha\v{r} model; Self-dual action; Hamiltonian formulation; GNH method.

%
%
\section{Introduction}{\label{sec_intro}}

Some approaches to quantum gravity rely on the Hamiltonian description of General Relativity (GR). Among them, the formulation found by Ashtekar \cite{Ashtekar:1986yd,Ashtekar:1987gu,Ashtekar:2004eh,Perez:2004hj}---upon which Loop Quantum Gravity has been built---displays some very interesting features originating in the geometrical meaning of its canonical variables. In terms of the Ashtekar variables, it is possible to build Hilbert spaces that seem appropriate to accommodate a working theory of quantum gravity. However, some difficulties have to be resolved before the quantization program for GR can be completed in this framework. In particular, the consistent quantization of the Hamiltonian constraint which encodes its non-trivial dynamics.

To disentangle the difficulties specifically associated with the Hamiltonian constraint from those coming from the absence of background spacetime structures (diffeomorphism invariance), some models have been proposed in the past that have the interesting feature of leading to Hamiltonian formulations in the Ashtekar phase space \textit{without} a Hamiltonian constraint. Among them, we would like to highlight the $2+1$ dimensional Husain model \cite{Husain:1992qx}, the Husain-Kucha\v{r} (HK) model \cite{Husain:1990vz} (and some extensions of it with matter \cite{Husain:1993he}) in $3+1$ dimensions and the Extended Husain-Kucha\v{r} (EHK) model presented in \cite{BarberoG:1997nrd}, where a time variable was incorporated in a natural way.

One of the goals of the present paper is to study in depth the dynamics of the EHK model and its relation with the standard one by disentangling the meaning of the Hamiltonian vector fields. The main technical tool that we will use for this purpose is the Gotay-Nester-Hinds (GNH) \cite{Gotay1,Gotay2} approach to the Hamiltonian dynamics of singular systems. The results of \cite{BarberoG:2023qih} suggest that suitable redefinitions of the field variables may provide useful insights into the dynamics. Moreover, they might help resolve some pending issues regarding gauge fixings and the possibility of building non-degenerate metrics, for which the use of a non-vanishing time variable is mandatory.

Although, in many cases, the results found by using the GNH method are equivalent to those obtained by relying on the ideas developed by Dirac for this same purpose \cite{nosdirac}, there are some differences. Some of them are quite obvious; for instance, the Hamiltonian formulations found by using Dirac's ``algorithm'' all live in the full phase space of the mechanical model (or field theory) in question. In contrast, those obtained using the GNH method live directly on the so-called \textit{primary constraint submanifold} $\mathsf{M}_0$. This implies that the symplectic form in Dirac's approach is the canonical one $\Omega$, whereas that relevant for the GNH formulation is the pullback of $\Omega$ onto $\mathsf{M}_0$, which may be degenerate. Despite this last apparent difficulty, the results obtained by using the GNH method can be quite useful and illuminating. For instance, it has been recently shown \cite{BarberoG:2023qih} how the Ashtekar formulation for Euclidean GR can be derived in a clean way from the self-dual action \cite{Samuel:1987td,Jacobson:1987yw} \textit{without relying on any gauge fixing}. This fact is difficult (though not impossible) to discover in Dirac's setting precisely because of the need to look at the pullback of $\Omega$ onto $\mathsf{M}_0$. Our purpose here is to explore what other interesting facts can be discovered for diff-invariant field theories by following the same philosophy and applying the lessons learned to the study of GR and its quantization. From this perspective, and despite the fact that the present paper may appear to be ``too technical'', there is a clear physical motivation behind it: exploit the natural variables naturally suggested by the application of the GNH method to find and understand interesting dynamical features of the dynamics of the models explored here. This is, in fact, a general idea that may be fruitful to study other field theories.

The layout of the paper is the following. After this introduction, we present in Section \ref{sec_Action} the actions, equations of motion and dynamical features of the field theories studied here. Before looking in detail at these models, we devote Section \ref{sec_Dirac_GNH} to a general discussion of the main features of the Dirac and GNH methods, in particular, the differences between them. Subsequently, we will present and discuss the Hamiltonian formulation for the Husain-Kucha\v{r} model in Section \ref{sec_Hamiltonian_HK} and the Extended Husain-Kucha\v{r} model in Section \ref{sec_Hamiltonian_extended}. After that, we will make some comments and discuss the main results in Section \ref{sec_conclusions}. We end the paper with several appendices. In Appendix \ref{appendix_notation}, we explain the notation that we use throughout the paper, Appendix \ref{appendix_identities} compiles several definitions and gives a number of identities needed to complete some of the computations mentioned in the paper, and Appendix \ref{appendix_math_details} contains the proofs of several mathematical results (presented in the form of lemmas) used throughout this work. Finally, Appendix \ref{appendix_details_EHK} provides some details on the implementation of the GNH method for the Extended Husain-Kucha\v{r} model.

%
%

\section{Dirac versus GNH: physical motivations}\label{sec_Dirac_GNH}

Dirac's approach to the Hamiltonian analysis of constrained systems---in particular gauge theories---appears in his classic \textit{Lectures on Quantum Mechanics} \cite{Dirac}. This fact alone gives away his main goal: to provide a procedure to quantize field theories for which there is not a suitable one-to-one correspondence between velocities and canonical momenta. For this purpose, it is very useful to keep the canonical symplectic structure used to define the basic Poisson brackets between the canonical variables $(q,p)$ because this facilitates the search for quantum operators $(\hat{q},\hat{p})$ whose commutators mimic the form of $\{q,p\}$.

Dirac developed a procedure to find a Hamiltonian defined in the full phase space $\Gamma$ (\textit{the total Hamiltonian} $H_T$) and a submanifold of $\Gamma$ where a consistent dynamics is defined by $H_T$. Here \textit{consistent} means that the constraints are preserved by the evolution defined by $H_T$ via Poisson brackets. After the constraints are classified as first or second class, they are quantized either directly in terms of Poisson brackets (if only first-class constraints show up) or with the help of Dirac brackets (if second-class constraints are also present). An important comment has to be made here: although from a classical perspective an important outcome of the method is the total Hamiltonian $H_T$ (with the multipliers introduced in its definition fixed by the consistency of the dynamics), most often Dirac's algorithm is applied as a step towards quantization. Thus, the details of the dynamics defined by $H_T$ often play a secondary role.

The goal of the GNH approach is also to find Hamiltonian formulations for singular Lagrangian systems. At face value, it seems to be quite different from Dirac's approach. Some of the most striking differences that can be gleaned from a first reading of \cite{Gotay1,Gotay2} are its geometric flavor and the emphasis on functional analytic issues when dealing with field theories (treated in a rather cavalier way in \cite{Dirac}). However, just by applying the geometric language of GNH to Dirac's approach (see \cite{nosdirac}), it is possible to see that both methods are actually quite close. For instance, there are no obstructions in principle to incorporate functional issues into Dirac's analysis. Also, dynamical consistency in Dirac's setting can be interpreted in terms of tangency conditions like in the GNH method. In many circumstances, this has a convenient consequence: the use of Poisson brackets (whose definition in terms of functional derivatives can be tricky \cite{BarberoG:2022lji}) can be circumvented.

At the end of the day, the most important difference between both methods is this: whereas the final outcome of Dirac's method is defined on the full phase space $\Gamma$, the Hamiltonian description that one gets in the GNH setting is defined on the primary constraint submanifold $\mathsf{M}_0$ (defined as the image of the fiber derivative $F\!L$ associated with a given Lagrangian $L$). Some comments are in order now. The first is that the relevant symplectic structures in both approaches are different: the canonical one $\Omega$ in Dirac's method and its pullback $\omega$ onto the primary constraint submanifold in the GNH approach. As a consequence, the GNH method seems to be less suited for quantization because $\omega$ is often degenerate and may look quite complicated if written in terms of the original canonical variables. A second comment is that, whereas one can get the GNH result from Dirac's method (pulling back everything onto $\mathsf{M}_0$), there is no compelling reason to do so as far as the main concern is quantization.

This notwithstanding, the GNH description can be quite useful. For instance, it has been shown in \cite{BarberoG:2023qih} that the form of $\omega$ actually suggests how to arrive at the Ashtekar formulation in a very natural way without the need for gauge fixing. One of the physical motivations of the present work is to apply the same procedure to other models to see how the mechanism discussed in \cite{BarberoG:2023qih} works and what can be learned from it. As we show later, there are, in fact, some subtle differences in the way the dynamics of the self-dual and EHK models work, in particular, the behavior of the time variable introduced in the EHK model and that of the 1-form field from which it comes. These will be discussed in the conclusions.

%
%

\section{Actions and equations of motion}\label{sec_Action}

\subsection{The Husain-Kucha\v{r} model}\label{subsec_HK}

Consider a closed (i.e. compact without boundary), orientable, and 3-dimensional manifold $\Sigma$ (hence, parallelizable) and the 4-manifold $\mathcal{M}=\mathbb{R}\times\Sigma$. $\mathcal{M}$ is foliated by the ``spatial sheets'' $\Sigma_\tau:=\{\tau\}\times\Sigma$ ($\tau\in\mathbb{R}$). Also, there is a canonical evolution vector field $\partial_t\in \mathfrak{X}(\mathcal{M})$, transverse to all $\Sigma_\tau$, defined by the tangent vectors to the curves $c_p:\mathbb{R}\rightarrow \mathcal{M}:\tau\mapsto (\tau,p)$. Finally, we have a scalar field $t\in C^\infty(\mathcal{M})$ defined as $t:\mathcal{M}\rightarrow \mathbb{R}:(\tau,p)\mapsto \tau$ and such that ${\partial_t}\iprod\bd t=1$.

We introduce the following geometric objects on $\mathcal{M}$:
\begin{align*}
&{\bm{\mathrm{A}}}^i\in\Omega^1(\mathcal{M})\,,&&i=1,2,3\,,\\
&{\bm{\mathrm{e}}}^i\in\Omega^1(\mathcal{M})\,,&&i=1,2,3\,,
\end{align*}
chosen in such a way that $(\bd t,{\bm{\mathrm{e}}}^i)$ is a coframe. As shown in Lemma \ref{lemma volume form}, this is equivalent to the condition that
\[
\frac{1}{3!}\varepsilon_{ijk}\bd t\wedge{\bm{\mathrm{e}}}^i\wedge{\bm{\mathrm{e}}}^j\wedge{\bm{\mathrm{e}}}^k=:\mathsf{vol}_t
\]
is a volume form. Notice that we have
\begin{equation}\label{formula_vol_1}
  \bd t\wedge {\bm{\mathrm{e}}}^i\wedge {\bm{\mathrm{e}}}^j\wedge {\bm{\mathrm{e}}}^k=\varepsilon^{ijk}\mathsf{vol}_t \,.
\end{equation}

The ${\bm{\mathrm{e}}}^i$ are linearly independent at every point of $\mathcal{M}$ but, as we only have three 1-forms in the 4-dimensional manifold $\mathcal{M}$, by themselves they do not define a coframe.

With the help of the fields introduced above, we define the covariant exterior differential ${\bm{\mathrm{D}}}$, which on the ${\bm{\mathrm{e}}}_i$ acts according to
\[
{\bm{\mathrm{D}}}{\bm{\mathrm{e}}}_i:=\bd{\bm{\mathrm{e}}}_i+\varepsilon_{ijk}{\bm{\mathrm{A}}}^j\wedge {\bm{\mathrm{e}}}^k\,,
\]
(generalized in the usual way to arbitrary differential forms) and the curvature 2-form
\[
{\bm{\mathrm{F}}}_i:=\bd{\bm{\mathrm{A}}}_i+\frac{1}{2}\varepsilon_{ijk}{\bm{\mathrm{A}}}^j\wedge{\bm{\mathrm{A}}}^k\,.
\]
In the previous expressions, $\bd$ denotes the exterior differential on $\mathcal{M}$.

The standard action for the Husain-Kucha\v{r} (HK) model \cite{Husain:1990vz} is
\begin{equation}\label{HK_action}
  S(\mathbf{e},\mathbf{A}) := \frac{1}{2} \int_{\mathcal{M}} \varepsilon_{ijk} \mathbf{e}^i \wedge \mathbf{e}^j \wedge \mathbf{F}^k\,.
\end{equation}
It is both diff-invariant\begin{align}\label{difftransformations}
  \begin{split}
  \delta {\mathbf{A}}^i& =\pounds_{\bm{\xi}}{\mathbf{A}}^i\,,\\
  \delta {\bm{\mathrm{e}}}^i& =\pounds_{\bm{\xi}}{\ee}^i\,,
  \end{split}
\end{align}
where $\bm{\xi}\in\mathfrak{X}(\mathcal{M})$, and invariant under the infinitesimal transformation
\begin{align}\label{SO3transformations}
  \begin{split}
  \delta {\mathbf{A}}^i& = {\bm{\mathrm{D}}}{\bm{\Lambda}}^i\,,\\
  \delta {\bm{\mathrm{e}}}^i& =\varepsilon^i_{\phantom{i}jk}{\ee}^j{\bm{\Lambda}}^k\,,
  \end{split}
\end{align}
where $({\bm{\Lambda}}^k)\in C^\infty(\mathcal{M})^3$ and ${\bm{\mathrm{D}}}{\bm{\Lambda}}^i:=\bm{\mathrm{d}}{\bm{\Lambda}}^i+\varepsilon^{ijk}{\bm{\mathrm{A}}}_j{\bm{\Lambda}}_k$. This allows us to interpret $i,j,k=1,2,3$ as ``internal $SO(3)$ indices''. In fact, they can be understood as abstract indices over the $SO(3)$ bundle (with the Killing metric), in which case $\varepsilon^i_{\phantom{i}jk}$ can be understood (with a bit of effort) as the structure constants of a graded Lie algebra with values in the space of forms. Finally, notice also that, using covariant phase space methods \cite{CPS}, we can prove that both infinitesimal symmetries \eqref{difftransformations} and \eqref{SO3transformations} are, in fact, gauge symmetries.

The field equations corresponding to the action \eqref{HK_action} are
\begin{align}\label{field_equations_HK}
\begin{split}
   & \varepsilon_{ijk}\bm{\mathrm{e}}^j\wedge {\bm{\mathrm{D}}} \bm{\mathrm{e}}^k=0\,, \\
   & \varepsilon_{ijk}\bm{\mathrm{e}}^j\wedge \bm{\mathrm{F}}^k=0\,.
   \end{split}
\end{align}
With the help of the volume form $\mathsf{vol}_t $ in $\mathcal{M}$ let us define
\begin{equation}\label{def_u}
\tilde{u}_0(\cdot):=\frac{1}{3!}\left(\frac{\cdot\wedge \varepsilon_{ijk}\bm{\mathrm{e}}^i\wedge \bm{\mathrm{e}}^j\wedge\bm{\mathrm{e}}^k}{\mathsf{vol}_t }\right)\,.
\end{equation}
At each point $p\in \mathcal{M}$, this is an element of $(T_p\mathcal{M})^{**}$ (i.e. a linear map from $(T_p\mathcal{M})^*$ to $\mathbb{R}$) and determines in a unique way a vector in $T_p\mathcal{M}$. By gathering these vectors for all $p$, we get a vector field $\tilde{u}_0\in\mathfrak{X}(\mathcal{M})$. It is straightforward to see that $\tilde{u}_0\iprod\bm{\mathrm{e}}^i=\tilde{u}_0(\bm{\mathrm{e}}^i)=0$ for $i=1,2,3$, and $\tilde{u}_0\iprod{\bm{\mathrm{d}}}t=1$.

\medskip

Let ${\bm{\xi}}^k\in\Omega^2(\mathcal{M})$, then, as shown in Lemma \ref{int_prods},  $\varepsilon_{ijk}\bm{\mathrm{e}}^j\wedge{\bm{\xi}}^k=0$ implies ${\tilde{u}_0}\iprod{\bm{\xi}}^k=0$. As a consequence of this, on solutions to the field equations \eqref{field_equations_HK}, we have
    \begin{equation}\label{conditions_u}
    \tilde{u}_0\iprod {\bm{\mathrm{D}}} \bm{\mathrm{e}}^i=0\,,\quad \tilde{u}_0\iprod\bm{\mathrm{F}}^i=0\,.
    \end{equation}

\medskip

Let us now take a solution $(\bm{\mathrm{e}}^i, \bm{\mathrm{A}}^i)$ to the field equations \eqref{field_equations_HK}, we then have
\begin{align*}
  &\pounds_{\tilde{u}_0}\bm{\mathrm{e}}^i=\tilde{u}_0\iprod{\bm{\mathrm{d}}}\bm{\mathrm{e}}^i+{\bm{\mathrm{d}}}(\tilde{u}_0\iprod\bm{\mathrm{e}}^i)
  =-\tilde{u}_0\iprod(\varepsilon^{ijk}{\bm{\mathrm{A}}}_j{\bm{\mathrm{e}}}_k)=\varepsilon^{ijk}{\bm{\mathrm{e}}}_j(\tilde{u}_0\iprod{\bm{\mathrm{A}}}_k)\\
  &\pounds_{\tilde{u}_0}\bm{\mathrm{A}}^{\!\!i}\!= \tilde{u}_0\iprod{\bm{\mathrm{d}}}\bm{\mathrm{A}}^i\!+\!{\bm{\mathrm{d}}}(\tilde{u}_0\iprod\bm{\mathrm{A}}^i)
  =-\frac{1}{2}\tilde{u}_0\iprod(\varepsilon^{ijk} {\bm{\mathrm{A}}}_j\wedge {\bm{\mathrm{A}}}_k)\!+\!{\bm{\mathrm{d}}}(\tilde{u}_0\iprod{\bm{\mathrm{A}}}^i)\\
  &\qquad\,\,={\bm{\mathrm{d}}}(\tilde{u}_0\iprod{\bm{\mathrm{A}}}^i)+\varepsilon^{ijk}{\bm{\mathrm{A}}}_j(\tilde{u}_0\iprod{\bm{\mathrm{A}}}_k)={\bm{\mathrm{D}}}(\tilde{u}_0\iprod{\bm{\mathrm{A}}}^i)
\end{align*}
where we have made use of Cartan's magic formula for the Lie derivative acting on differential forms $\beta$ ($\pounds_{\tilde{u}_0}\beta=\tilde{u}_0\iprod{\bm{\mathrm{d}}}\beta+{\bm{\mathrm{d}}}\tilde{u}_0\iprod\beta$), the conditions \eqref{conditions_u}, and $\tilde{u}_0\iprod{\bm{\mathrm{e}}}^k=0$. As we can see, Lie-dragging solutions to the field equations along the vector field $\tilde{u}_0$ is equivalent to performing the infinitesimal transformations \eqref{SO3transformations} with local parameter ${\bm{\Lambda}}^i:=\tilde{u}_0\iprod{\bm{\mathrm{A}}}^i$. These transformations leave the degenerate 4-metric ${\bm{\mathrm{e}}}^i\otimes {\bm{\mathrm{e}}}_i$ invariant.

\medskip

As a final comment, it is important to point out that, as proved in Lemma \ref{useful_result}, the solutions to the field equations \eqref{field_equations_HK} also satisfy
\[
\DD\ee_i\!\wedge\FF^i=0\,.
\]

\subsection{The Extended Husain-Kucha\v{r} model}\label{subsec_EH}

Now we will extend the HK model along the lines of \cite{BarberoG:1997nrd} by introducing an action that somehow ``interpolates'' between the usual HK action and the one for self-dual gravity \cite{BarberoG:1994fcp,BarberoG:2023qih}
\begin{equation}\label{EHK_action}
  S({\bm{\mathrm{e}}},{\bm{\mathrm{A}}}, {\bm{\phi}}):=\int_{\mathcal{M}} \left(\frac{1}{2}\varepsilon_{ijk}{\bm{\mathrm{e}}}^i\wedge {\bm{\mathrm{e}}}^j\wedge {\bm{\mathrm{F}}}^k-{\bm{\mathrm{d}}}{\bm{\phi}}\wedge {\bm{\mathrm{e}}}_i\wedge{\bm{\mathrm{F}}}^i\right)\,.
\end{equation}
Here, the manifold $\mathcal{M}$ and the fields used to write \eqref{EHK_action} are the same as those introduced in the preceding subsection (and satisfy the same conditions), but we have introduced an additional scalar field ${\bm{\phi}}\in C^\infty(\mathcal{M})$.

The action \eqref{EHK_action} is a particular case of the action for Euclidean self-dual gravity discussed in \cite{BarberoG:2023qih}
\begin{equation}\label{Self_dual_action}
  S({\bm{\mathrm{e}}},{\bm{\mathrm{A}}}, {\bm{\alpha}}):=\int_{\mathcal{M}} \left(\frac{1}{2}\varepsilon_{ijk}{\bm{\mathrm{e}}}^i\wedge {\bm{\mathrm{e}}}^j\wedge {\bm{\mathrm{F}}}^k-{\bm{\alpha}}\wedge {\bm{\mathrm{e}}}_i\wedge{\bm{\mathrm{F}}}^i\right)\,,
\end{equation}
with ${\bm{\alpha}}$ exact (i.e. ${\bm{\alpha}}={\bm{\mathrm{d}}}{\bm{\phi}}$). As in the previous subsection, \eqref{EHK_action} is diff-invariant and invariant under the infinitesimal $SO(3)$ transformations of the HK model \eqref{SO3transformations}.

 The field equations given by \eqref{EHK_action} are
\begin{align}\label{field_equations_EHK}
\begin{split}
   & \varepsilon_{ijk}\bm{\mathrm{e}}^j\wedge {\bm{\mathrm{D}}} \bm{\mathrm{e}}^k-{\bm{\mathrm{d}}}{\bm{\phi}}\wedge{\bm{\mathrm{D}}}\bm{\mathrm{e}}_i=0\,, \\
   & \varepsilon_{ijk}\bm{\mathrm{e}}^j\wedge \bm{\mathrm{F}}^k+{\bm{\mathrm{d}}}{\bm{\phi}}\wedge \bm{\mathrm{F}}_i=0\,,\\
   & {\bm{\mathrm{D}}} \bm{\mathrm{e}}^i\!\wedge \bm{\mathrm{F}}_i=0\,.
\end{split}
\end{align}
It is interesting to notice now that, as a consequence of Lemma \ref{useful_result}, the last equation in \eqref{field_equations_EHK} can be ignored since it is satisfied if the first two ones hold.

%
%

\section{GNH formulation for the HK model}\label{sec_Hamiltonian_HK}

To get the Lagrangian and Hamiltonian formulations of action \eqref{HK_action}, we take advantage of the foliation naturally associated with $\mathcal{M}=\mathbb{R}\times\Sigma$ discussed above (note that $\jmath_t(\Sigma)=\Sigma_t$). By making use of
\[
\int_{\mathbb{R}\times\Sigma}\mathcal{L}=\int_{\mathbb{R}\times\Sigma}\mathrm{d}t\wedge\partial_t\iprod\mathcal{L}=\int_{\mathbb{R}}\mathrm{d}t\int_{\Sigma_t}\partial_t\iprod\mathcal{L}
=\int_{\mathbb{R}}\mathrm{d}t\int_{\Sigma}\jmath_t^*\partial_t\iprod\mathcal{L}\,.
\]
we get the Lagrangian $L:=\int_\Sigma\jmath_t^*\partial_t\iprod\mathcal{L}$ defined on $\Sigma$ from the 4-form $\mathcal{L}$ appearing in the action. This Lagrangian is defined on the tangent bundle of the configuration space
\[
Q=C^\infty(\Sigma)^3 \times \Omega^1(\Sigma)^3 \times C^\infty(\Sigma)^3 \times \Omega^1(\Sigma)^3\,,
\]
subject to the condition that the $(e^i)$ define a coframe on $\Sigma$. The points of $Q$ have the form $(e_{\mathrm{t}}^i,e^i,A_{\mathrm{t}}^i,A^i)$. To arrive at this result, we interpret the objects
\begin{align*}
&  e_{\mathrm{t}}^i(t)\,:=\jmath_t^\ast\partial_t\iprod\bm{\mathrm{e}}^i\,, && e^i(t)\,:=\jmath_t^\ast \bm{\mathrm{e}}^i\,,\\
&  A_{\mathrm{t}}^i(t):=\jmath_t^\ast\partial_t\iprod\Aa^i\,, && A^i(t):=\jmath_t^\ast \Aa^i\,,
\end{align*}
as defining curves in the configuration space $Q$ and consider also the velocities
\begin{align*}
  & v_{e_{\mathrm{t}}}^i(t):=\jmath_t^\ast \pounds_{\partial_t} ({\partial_t}\iprod{\bm{\mathrm{e}}}^i)=\frac{\mathrm{d}e_{\mathrm{t}}^i}{\mathrm{d}\tau}(t)\,, && v_{e}^i(t):=\jmath_t^\ast \pounds_{\partial_t} \bm{\mathrm{e}}^i=\frac{\mathrm{d}e^i}{\mathrm{d}\tau}(t)\,,\\
  & v_{A_{\mathrm{t}}}^i(t):=\jmath_t^\ast \pounds_{\partial_t} ({\partial_t}\iprod{\Aa}^i)=\frac{\mathrm{d}A_{\mathrm{t}}^i}{\mathrm{d}\tau}(t)\,, && v_A^i(t):=\jmath_t^\ast \pounds_{\partial_t} \Aa^i=\frac{\mathrm{d}A^i}{\mathrm{d}\tau}(t)\,,
\end{align*}
defined in terms of $\pounds_{\partial_t}$,  the Lie derivative along $\partial_t$.

The Lagrangian is
\begin{equation*}
    L(\mathrm{v}) = \int_\Sigma \left(\frac{1}{2} \varepsilon_{ijk} e^i \wedge e^j \wedge v_A^k -  A_{\mathrm{t}}^i \varepsilon_{ijk} e^j \wedge D e^k + e_{\mathrm{t}}^i \varepsilon_{ijk} e^j \wedge F^k \right)\,,
\end{equation*}
where the velocity corresponding to $(e_{\mathrm{t}}^i,e^i,A_{\mathrm{t}}^i,A^i)$ is $\mathrm{v}:=(v_{e_\mathrm{t}}^i, v_e^i,v_{A_\mathrm{t}}^i, v_A^i)$. We have defined
\[
F^i:=\mathrm{d}A^i+\frac{1}{2}\varepsilon^{ijk}A_j\wedge A_k,\qquad\qquad
De^i:=\mathrm{d}e^i+\varepsilon^{ijk}A_j\wedge e_k\,.
\]

The fiber derivative is
\begin{equation}\label{FL_st}
    \langle F\!L(\mathrm{v}),\mathrm{w}\rangle=\frac{1}{2}\int_\Sigma \varepsilon_{ijk} e^i \wedge e^j \wedge w_A^k \,,
\end{equation}
and the primary constraint submanifold in phase space is defined by

\vspace*{-3mm}

\begin{alignat*}{3}
    &\mathbf{p}_{e_{\mathrm{t}}} &&= 0 \,, \\
    &\mathbf{p}_e &&= 0 \,, \\
    &\mathbf{p}_{A_{\mathrm{t}}} &&= 0 \,, \\
    &\mathbf{p}_A (\mathrm{w}) &&= \frac{1}{2}\int_\Sigma \varepsilon_{ijk} e^j \wedge e^k \wedge w_A^i \,,
\end{alignat*}
where $\mathrm{w}:=(w_{e_\mathrm{t}}^i, w_e^i,w_{A_\mathrm{t}}^i, w_A^i)$ and we denote the momenta at the fiber of $T^*Q$ over the point $(e_{\mathrm{t}}^i,e^i,A_{\mathrm{t}}^i,A^i)\in Q$ as $\mathbf{p}:=(\mathbf{p}_{e_{\mathrm{t}}},\mathbf{p}_e,\mathbf{p}_{A_{\mathrm{t}}},\mathbf{p}_A)$.

On the primary constraint submanifold, the Hamiltonian is given by $H=E\circ F\!L^{-1}$ where $E:=\langle F\!L(v),v\rangle-L$ is the energy (notice that, as $F\!L$ is injective, it is a bijection onto its image). $H$ is given by
\begin{equation}
\label{asd_hamiltonian_st}
     H(\mathbf{p}) = \int_\Sigma \left( A_{\mathrm{t}}^i \varepsilon_{ijk} e^j \wedge D e^k - e_{\mathrm{t}}^i \varepsilon_{ijk} e^j \wedge F^k \right)\,.
\end{equation}
As we can see, the Hamiltonian depends on $\mathbf{p}$ only through its base point.

\medskip

A vector field in phase space $\mathbb{Y}\in\mathfrak{X}(T^*Q)$ has the form
\[
\mathbb{Y}=(Y_{e_{\mathrm{t}}}^i, Y_e^i, Y_{A_{\mathrm{t}}}^i, Y_A^i, {\bm{\mathrm{Y}}}_{\!\!\bm{\mathrm{p}}_{e_{\mathrm{t}}}}, {\bm{\mathrm{Y}}}_{\!\!\bm{\mathrm{p}}_{e}}, {\bm{\mathrm{Y}}}_{\!\!\bm{\mathrm{p}}_{A_{\mathrm{t}}}}, {\bm{\mathrm{Y}}}_{\!\!\bm{\mathrm{p}}_A})\,,
\]
where we use boldface characters for the components of $\mathbb{Y}$ that are dual objects. Notice that $(Y_e^i)\,,(Y_A^i)\in \Omega^1(\Sigma)^3$ and $(Y_{e_{\mathrm{t}}}^i)\,,(Y_{A_{\mathrm{t}}}^i)\in C^\infty(\Sigma)^3$. The vector fields $\mathbb{Y}\in \mathfrak{X}(T^*Q)$ tangent to the primary constraint submanifold $\mathsf{M}_0$ have the following form
\[
\mathbb{Y}=(Y_{e_{\mathrm{t}}}, Y_e, Y_{A_{\mathrm{t}}}, Y_A,{\bm{0}},{\bm{0}},{\bm{0}},{\bm{\mathrm{Y}}}_{\!\!\bm{\mathrm{p}}_{A}}),
\]
where
\[
{\bm{\mathrm{Y}}}_{\!\!\bm{\mathrm{p}}_{A}}(\cdot)=\int_\Sigma \varepsilon_{ijk}Y_e^j\wedge e^k\wedge \cdot\,.
\]

We can write the action of the exterior differential in phase space $\dd$ of the Hamiltonian $H$ acting on a vector field $\mathbb{Y}$ as
\begin{align}\label{dH_st}
\begin{split}
  \dd H(\mathbb{Y})=\int_\Sigma\Big(&-Y_{e_{\mathrm{t}}}^i\varepsilon_{ijk}e^j\wedge F^k\\
   & +Y_e^i\wedge \big(\varepsilon_{ijk}D\!A_{\mathrm{t}}^j\wedge e^k+\varepsilon_{ijk}e_{\mathrm{t}}^jF^k\big)\\
   & +Y_{A_{\mathrm{t}}}^i \varepsilon_{ijk}e^j\wedge De^k \\
   & +Y_A^i\wedge\big(D(\varepsilon_{ijk}e^je_{\mathrm{t}}^k)-A_{\mathrm{t}}^je_j\wedge e_ i\big) \Big)\,.
\end{split}
\end{align}

\medskip

The action of the canonical symplectic form $\Omega$ on vector fields in phase space $\mathbb{Y}, \mathbb{Z}\in \mathfrak{X}(T^*Q)$
\begin{align*}
  \Omega(\mathbb{Z},\mathbb{Y}) = & +{\bm{\mathrm{Y}}}_{\!\!\bm{\mathrm{p}}_{e_{\mathrm{t}}}}(Z_{e_{\mathrm{t}}})  -  {\bm{\mathrm{Z}}}_{\bm{\mathrm{p}}_{e_{\mathrm{t}}}}(Y_{e_{\mathrm{t}}})
                                  +{\bm{\mathrm{Y}}}_{\!\!\bm{\mathrm{p}}_{e}}(Z_{e})  -  {\bm{\mathrm{Z}}}_{\bm{\mathrm{p}}_{e}}(Y_{e})\\
                                & +{\bm{\mathrm{Y}}}_{\!\!\bm{\mathrm{p}}_{A_{\mathrm{t}}}}\!(Z_{A_{\mathrm{t}}})  -  \!{\bm{\mathrm{Z}}}_{\bm{\mathrm{p}}_{A_{\mathrm{t}}}}\!(Y_{A_{\mathrm{t}}}\!)+\!{\bm{\mathrm{Y}}}_{\!\!\bm{\mathrm{p}}_{A}}\!(Z_{A})\!  -  {\bm{\mathrm{Z}}}_{\bm{\mathrm{p}}_{A}}\!(Y_{A})\,.
\end{align*}

Instead of using the full phase space, in the GNH approach one works directly on the primary constraint submanifold $\mathsf{M}_0$. In the models we are considering here, this is straightforward because $\mathsf{M}_0$ can be easily identified with $Q$ (the momenta are either zero or can be written in terms of objects of $Q$). Thus, the vector fields on $\mathsf{M}_0$ have the form
\[
\mathbb{Y}_0=(Y_{e_{\mathrm{t}}}^i, Y_e^i, Y_{A_{\mathrm{t}}}^i, Y_A^i)\,,
\]
and the pullback of $\Omega$ to $\mathsf{M}_0$ acting on these vector fields is
\begin{equation}\label{pullback_omega_st}
\omega(\mathbb{Z}_0,\mathbb{Y}_0)=\int_\Sigma\,\left(Y_e^i\wedge\varepsilon_{ijk}e^j\wedge Z_A^k-Y_A^i\wedge\varepsilon_{ijk}e^j\wedge Z_e^k\right)\,.
\end{equation}
We can now solve the fundamental equation of the GNH approach $\imath_{\mathbb{Z}_0}\omega=\dd H$ by equating the terms proportional to the different components of the vector field $\mathbb{Y}$ in \eqref{dH_st} and \eqref{pullback_omega_st}. By doing this, we get the secondary constraints
\begin{align}
&\varepsilon_{ijk}e^j\wedge F^k=0\,,\label{L1_st}\\
&\varepsilon_{ijk}e^j\!\wedge\! De^k=0\,,\label{L2_st}
\end{align}
and the following equations for the components of the Hamiltonian vector field $\mathbb{Z}_0$
\begin{align}
  & \varepsilon_{ijk}e^j\wedge(Z_A^k-D\!A_{\mathrm{t}}^k)=\varepsilon_{ijk}e_{\mathrm{t}}^jF^k\,,\label{E1_st}\\
  & \varepsilon_{ijk}e^j\wedge(Z_e^k-De_{\mathrm{t}}^k-\varepsilon^k_{\phantom{k}\ell m}e^\ell A_{\mathrm{t}}^m)=\varepsilon_{ijk}e_{\mathrm{t}}^jDe^k.\label{E2_st}
\end{align}
In terms of the objects
\begin{align}
X_A^i&:=Z_A^i-DA_{\mathrm{t}}^i\,,\\
X_e^i&:=Z_e^i-De_{\mathrm{t}}^i-\varepsilon^i_{\phantom{i}jk}e^jA_{\mathrm{t}}^k\,,
\end{align}
the preceding equations become
\begin{align}
  & \varepsilon_{ijk}e^j\wedge X_A^k=\varepsilon_{ijk}e_{\mathrm{t}}^jF^k\,,\label{EE1_st}\\
  & \varepsilon_{ijk}e^j\wedge X_e^k=\varepsilon_{ijk}e_{\mathrm{t}}^jDe^k.\label{EE2_st}
\end{align}
There are no conditions on $Z^i_{e_{\mathrm{t}}}$ and $Z^i_{A_{\mathrm{t}}}$ so, at this point, they are arbitrary. This implies that $e_{\mathrm{t}}^i$ and $A_{\mathrm{t}}^i$ are also arbitrary.

We have to check now if the vector fields satisfying \eqref{E1_st} and \eqref{E2_st} are tangent to the submanifold $\mathsf{M}_1\subset\mathsf{M}_0$ defined by the secondary constraints \eqref{L1_st} and \eqref{L2_st}. To do this, we solve \eqref{E1_st} and \eqref{E2_st} for $Z_e^i$ and $Z_A^i$.

To understand the content of the secondary constraints \eqref{L1_st} and \eqref{L2_st} (and also, for notational convenience) we introduce
\[
\mathbb{F}_{ij}:=\left(\frac{F_i\wedge e_j}{\mathsf{vol}_e}\right)\,,\quad \mathbb{B}_{ij}:=\left(\frac{De_i\wedge e_j}{\mathsf{vol}_e}\right)\,,
\]
where we have made use of the fact that $\mathsf{vol}_e:=\frac{1}{3!}\varepsilon_{ijk}e^i\wedge e^j\wedge e^k$ is a volume form on $\Sigma$. In terms of $\mathbb{F}_{ij}$ and $\mathbb{B}_{ij}$ we have
\[
F_i=\frac{1}{2}\mathbb{F}_{ij}\varepsilon^{jk\ell}e_k\wedge e_\ell\,,\quad De_i=\frac{1}{2}\mathbb{B}_{ij}\varepsilon^{jk\ell}e_k\wedge e_\ell\,,
\]
and we can write the constraints \eqref{L1_st} and \eqref{L2_st} in the following simple form
\begin{align}
  & \varepsilon_{ijk}\mathbb{F}^{jk}=0\,, \label{LL1_st}\\
  & \varepsilon_{ijk}\mathbb{B}^{jk}=0\,. \label{LL2_st}
\end{align}

Equations \eqref{EE1_st}-\eqref{EE2_st} can be solved for $X_{A}^i$ and $X_{e}^i$ by following the procedure explained in Appendix C of \cite{BarberoG:2021ekv}, obtaining
\begin{align}\label{X_st}
\begin{split}
  X_A^i & =\mathbb{F}^{ij}\varepsilon_{jk\ell}e_{\mathrm{t}}^ke^\ell\,, \\
  X_e^i & =\mathbb{B}^{ij}\varepsilon_{jk\ell}e_{\mathrm{t}}^ke^\ell\,,
\end{split}
\end{align}
or, equivalently,
\begin{align}\label{Z_st}
\begin{split}
  Z_A^i & =DA_{\mathrm{t}}^i+\mathbb{F}^{ij}\varepsilon_{jk\ell}e_{\mathrm{t}}^ke^\ell\,, \\
  Z_e^i & =De_{\mathrm{t}}^i+\varepsilon^{ijk}e_jA_{\mathrm{t}k}+\mathbb{B}^{ij}\varepsilon_{jk\ell}e_{\mathrm{t}}^ke^\ell\,.
\end{split}
\end{align}
From these expressions, it is possible to get a very neat understanding of the dynamics of the model. To this end, we point out that the non-degeneracy of the triads $e^i$ allows us to find the unique vector field $\xi\in\mathfrak{X}(\Sigma)$ satisfying the condition $\xi\iprod e^i=e_{\mathrm{t}}^i$. In terms of this geometric object and using Cartan's magic formula, we have
\begin{align}\label{ZZ_st}
\begin{split}
  Z_A^i & =\pounds_\xi A^i+D(A_{\mathrm{t}}^i-\xi\iprod A^i)\,, \\
  Z_e^i & =\pounds_\xi e^i-\varepsilon^i_{\phantom{i}jk}(A_{\mathrm{t}}^j-\xi\iprod A^j)e^k\,.
\end{split}
\end{align}
As we can see, we recover a well-known result: the dynamics is simply given by diffeomorphisms generated by the vector field $\xi$ and field-dependent internal $SO(3)$ transformations.

The consistency of the dynamics in the GNH approach is proven by showing that the vector fields just found are tangent to $\mathsf{M}_1$, which we recall is defined by the secondary constraints \eqref{L1_st} and \eqref{L2_st}. The relevant tangency conditions are
\begin{align}
  & \varepsilon_{ijk}Z_e^j\wedge F^k+\varepsilon_{ijk}e^j\wedge DZ_A^k=0\,,\label{tan1_st} \\
  & D(\varepsilon_{ijk}e^j\wedge Z_e^k)+Z_A^j\wedge e_i\wedge e_j=0\,.\label{tan2_st}
\end{align}
By following a procedure similar to the one used in the tangency analysis presented in \cite{BarberoG:2023qih}, it is possible to remove the derivatives of $Z_A^i$ and $Z_e^i$ from \eqref{tan1_st} and \eqref{tan2_st} and rewrite these tangency conditions in the form
\begin{align}
  & \varepsilon_{ijk}X_e^j\wedge F^k+\varepsilon_{ijk}De^j\wedge X_A^k=0\,,\label{tan01_st}\\
  & e_i\wedge e_j\wedge X_A^j-e_{\mathrm{t}j}e_i\wedge F^j+e_{\mathrm{t}i} F_j\wedge e^j=0\,.\label{tan02_st}
\end{align}
By plugging \eqref{X_st} into these expressions, it is straightforward to check that they hold modulo the secondary constraints.

Although, at this point, we have arrived at a satisfactory understanding of the dynamics of the model within the GNH framework, it is desirable to connect the previous results with the literature on the subject and, in particular, to put them in the language of the Ashtekar formulation for GR (of which the HK model can be considered as a precursor). To this end, we point out that the pullback $\omega$ of the canonical symplectic form onto the primary constraint submanifold $\mathsf{M}_0$ can be formally written as
\begin{equation}\label{omega_st}
\omega=\int_\Sigma \dd A^i\ww\dd\left(\frac{1}{2}\varepsilon_{ijk}e^j\wedge e^k\right)\,,
\end{equation}
which leads to \eqref{pullback_omega_st}. This suggests the introduction of new variables that will allow us to rewrite \eqref{omega_st} in terms of canonically conjugate objects in the usual sense. Namely, we define the 2-forms $(E_i)\in\Omega^2(\Sigma)^3$
\[
E_i:=\frac{1}{2}\varepsilon_{ijk}e^j\wedge e^k\,.
\]
We introduce now a fiducial volume form  $\mathsf{vol}_0$ in $\Sigma$ and, by duality, we define the following vector fields
\[
\widetilde{E}_i:=\left(\frac{\cdot\wedge E_i}{\mathsf{vol}_0}\right)\in \mathfrak{X}(\Sigma)\,.
\]

In terms of the basic geometric objects $(e_{\mathrm{t}}^i,\widetilde{E}^i,A_{\mathrm{t}}^i,A^i)$ and the associated vector fields $(Y_{e_{\mathrm{t}}}^i, Y_{\widetilde{E}}^i, Y_{A_{\mathrm{t}}}^i, Y_A^i)$, we can write the pre-symplectic form as

\[
\omega=\int_\Sigma \left(\dd A^i\ww\dd \widetilde{E}_i\right) \mathsf{vol}_0\,.
\]
To understand this expression, the following computation is helpful
\begin{align*}
\imath_\mathbb{Y}\imath_{\mathbb{X}}\omega&=\imath_{\mathbb{Y}}\int_\Sigma\left(X_A^i\wedge \dd E_i-\dd A^i\wedge X_{Ei}\right)=\int_\Sigma \left(X_{A}^i\wedge Y_{Ei}-Y_{A}^i\wedge X_{Ei}\right)\\
&= \hphantom{\imath_{\mathbb{Y}}} \int_\Sigma \left(X_{A}^i\wedge (Y_{\widetilde{E}i}\iprod\mathsf{vol}_0)-Y_{Ai}\wedge (X_{\widetilde{E}}^i\iprod\mathsf{vol}_0)\right)=\int_\Sigma \left(Y_{\widetilde{E}i}\iprod X_A^i-X_{\widetilde{E}i}\iprod Y_A^i\right)\mathsf{vol}_0\,,
\end{align*}
where we have used the fact that, for 1-forms $Z_A^i$, one has
\[
0=Z_A^i\wedge \mathsf{vol}_0\Rightarrow 0=Y_{\widetilde{E}i}\iprod(Z_A^i\wedge \mathsf{vol}_0)=(Y_{\widetilde{E}i}\iprod Z_A^i)\mathsf{vol}_0-Z_A^i\wedge (Y_{\widetilde{E}i}\iprod\mathsf{vol}_0)\,.
\]
It is important to note at this point that $\omega$ is just pre-symplectic because it involves neither $e_{\mathrm{t}}^i$ nor $A_{\mathrm{t}}^i$ (which can then be taken as ``external'' inputs).

In terms of the variables defined above, the constraints become (see \cite{BarberoG:2023qih})
\begin{align}
  & \mathrm{div}_0\widetilde{E}_i+\varepsilon_{ijk}{\widetilde{E}^k}\iprod A^j=0\,,\label{const_st_Ei1} \\
  & {\widetilde{E}_i}\iprod F^i=0\,.\label{const_st_Ei2}
\end{align}

Finally, the Hamiltonian vector fields are
\begin{align}\label{vec_st_Ei}
  & Z_A^i =\pounds_\xi A^i+D(A_{\mathrm{t}}^i-\xi\iprod A^i)\,, \\
  & Z_{\widetilde{E}}^i =\pounds_\xi \widetilde{E}^i+(\mathrm{div}_0\xi) \widetilde{E}^i-\varepsilon^{ijk}(A_{\mathrm{t}j}-\xi\iprod A_j)\widetilde{E}_k\,,
\end{align}
where, as above, $\xi\in\mathfrak{X}(\Sigma)$ satisfies $\xi\iprod e^i=e_{\mathrm{t}}^i$ and $\pounds_\xi \widetilde{E}^i$ is the standard expression for the Lie derivative of a vector field.

To wrap up this section we want to emphasize that we arrive in a very natural way at the Ashtekar phase space: it suffices to look at \eqref{omega_st} to identify the variable $E_i$ which gives rise to $\widetilde{E}_i$ after realizing that the introduction of a fiducial volume form $\mathsf{vol}_0$ leads to the simple form of the constraints \eqref{const_st_Ei1} and \eqref{const_st_Ei2}. Notice, by the way, that this is usually done in a local coordinate patch and relying on the coordinate volume form $\mathrm{d}x_1\wedge\mathrm{d}x_2\wedge\mathrm{d}x_3$.

%
%

\section{GNH formulation for the Extended HK model}\label{sec_Hamiltonian_extended}

To get the Lagrangian and Hamiltonian formulations given by the action \eqref{EHK_action} we proceed as in the HK model studied above. The configuration space is now
\[
Q=C^\infty(\Sigma)^3 \times \Omega^1(\Sigma)^3 \times C^\infty(\Sigma)^3 \times \Omega^1(\Sigma)^3\times C^\infty(\Sigma)\,,
\]
with the $(e^i)$ defining a coframe on $\Sigma$. The points of $Q$ have the form $(e_{\mathrm{t}}^i,e^i,A_{\mathrm{t}}^i,A^i,\phi)$. The Lagrangian is
\begin{align*}
    L(\mathrm{v}) = &\int_\Sigma \left( \Big(\frac{1}{2} \varepsilon_{ijk} e^i \wedge e^j +e_k\wedge\mathrm{d}\phi \Big) \wedge v_A^k +  A_{\mathrm{t}}^i D \Big( \frac{1}{2}\varepsilon_{ijk} e^j \wedge e^k +  e_i\wedge \mathrm{d}\phi \Big)\right.\\
    & \hspace{5.5cm}\left.{\phantom{\Big|}}-(e_i \wedge F^i)v_\phi + e_{\mathrm{t}}^i \left( \varepsilon_{ijk} e^j \wedge F^k + F^i\wedge \mathrm{d}\phi \right) \right)\,,
\end{align*}
where the velocity corresponding to $(e_{\mathrm{t}}^i,e^i,A_{\mathrm{t}}^i,A^i,\phi)$ is $\mathrm{v}:=(v_{e_\mathrm{t}}^i, v_e^i,v_{A_\mathrm{t}}^i, v_A^i,v_\phi)$. The fiber derivative is
\begin{equation}\label{FL}
    \langle F\!L(\mathrm{v}),\mathrm{w}\rangle=\int_\Sigma \left(\Big(\frac{1}{2}\varepsilon_{ijk} e^i \wedge e^j + e_k \wedge \mathrm{d}\phi\Big) \wedge w_A^k- (e_i\wedge F^i)w_\phi \right) \,,
\end{equation}
and the primary constraint submanifold in phase space is defined by

\vspace*{-3mm}

\begin{alignat*}{3}
    &\mathbf{p}_{e_{\mathrm{t}}} &&= 0 \,, \\
    &\mathbf{p}_e &&= 0 \,, \\
    &\mathbf{p}_{A_{\mathrm{t}}} &&= 0 \,, \\
    &\mathbf{p}_A (\mathrm{w}) &&= \int_\Sigma \left( \frac{1}{2}\varepsilon_{ijk} e^j \wedge e^k  +  e_i \wedge \mathrm{d}\phi \right) \wedge w_A^i \,, \\
    &\mathbf{p}_\phi (\mathrm{w}) &&= -\int_\Sigma (e_i\wedge F^i)w_\phi \,,
\end{alignat*}
where $\mathrm{w}:=(w_{e_\mathrm{t}}^i, w_e^i,w_{A_\mathrm{t}}^i, w_A^i,w_\phi)$.

The Hamiltonian is now
\begin{equation}
\label{asd_hamiltonian}
     H(\mathbf{p}) = -\int_\Sigma \left( A_{\mathrm{t}}^i D \Big( \frac{1}{2}\varepsilon_{ijk} e^j \wedge e^k +e_i\wedge \mathrm{d}\phi\Big) + e_{\mathrm{t}}^i \Big( \varepsilon_{ijk} e^j \wedge F^k + F_i\wedge \mathrm{d}\phi \Big) \right)\,.
\end{equation}
In the previous expression, we denote the momenta at the fiber of $T^*Q$ over the point $(e_{\mathrm{t}}^i,e^i,A_{\mathrm{t}}^i,A^i,\phi)\in Q$ as $\mathbf{p}:=(\mathbf{p}_{e_{\mathrm{t}}},\mathbf{p}_e,\mathbf{p}_{A_{\mathrm{t}}},\mathbf{p}_A,\mathbf{p}_\phi)$. The Hamiltonian depends on $\mathbf{p}$ only through its base point.

\medskip

A vector field in phase space $\mathbb{Y}\in\mathfrak{X}(T^*Q)$ has the structure
\[
\mathbb{Y}=(Y_{e_{\mathrm{t}}}^i, Y_e^i, Y_{A_{\mathrm{t}}}^i, Y_A^i, Y_\phi, {\bm{\mathrm{Y}}}_{\!\!\bm{\mathrm{p}}_{e_{\mathrm{t}}}}, {\bm{\mathrm{Y}}}_{\!\!\bm{\mathrm{p}}_{e}}, {\bm{\mathrm{Y}}}_{\!\!\bm{\mathrm{p}}_{A_{\mathrm{t}}}}, {\bm{\mathrm{Y}}}_{\!\!\bm{\mathrm{p}}_A}, {\bm{\mathrm{Y}}}_{\!\!\bm{\mathrm{p}}_{\phi}})\,,
\]
where we use boldface characters for the components of $\mathbb{Y}$ that are dual objects (covectors in phase space). Notice that $Y_e^i\,,Y_A^i\in \Omega^1(\Sigma)$ and $Y_{e_{\mathrm{t}}}^i\,,Y_{A_{\mathrm{t}}}^i\,,Y_\phi\in C^\infty(\Sigma)$.

We can now write the action of the exterior differential in phase space $\dd$ of the Hamiltonian $H$ acting a vector field $\mathbb{Y}$ as
\begin{align}
  \dd H(\mathbb{Y})=\int_\Sigma\Big(&-Y_{e_{\mathrm{t}}}^i(\varepsilon_{ijk}e^j\wedge F^k+\mathrm{d}\phi\wedge F_i)\nonumber  \\
   & +Y_e^i\wedge \big(\varepsilon_{ijk}D\!A_{\mathrm{t}}^j\wedge e^k-D\!A_{\mathrm{t}i}\wedge \mathrm{d}\phi+\varepsilon_{ijk}e_{\mathrm{t}}^jF^k\big)\nonumber \\
   & +Y_{A_{\mathrm{t}}}^i D(\mathrm{d}\phi\wedge e_ i-\frac{1}{2}\varepsilon_{ijk}e^j\wedge e^k) \label{dH}\\
   & +Y_A^i\wedge\big(D(\varepsilon_{ijk}e_je_{\mathrm{t}}^k- e_{\mathrm{t}}^i\mathrm{d}\phi)-A_{\mathrm{t}}^je_j\wedge e_ i-\varepsilon_{ijk}A_{\mathrm{t}}^j\mathrm{d}\phi\wedge e^k\big) \nonumber\\
   & +Y_\phi \,\mathrm{d}\big(e_{\mathrm{t}}^iF_i+A_{\mathrm{t}}^iDe_i\big)\Big)\,.\nonumber
\end{align}

\medskip

The action of the canonical symplectic form $\Omega$ on $\mathbb{Y}, \mathbb{Z}\in \mathfrak{X}(T^*Q)$ on vector fields in phase space is
\begin{align*}
  \Omega(\mathbb{Z},\mathbb{Y}) = & +{\bm{\mathrm{Y}}}_{\!\!\bm{\mathrm{p}}_{e_{\mathrm{t}}}}(Z_{e_{\mathrm{t}}})  -  {\bm{\mathrm{Z}}}_{\bm{\mathrm{p}}_{e_{\mathrm{t}}}}(Y_{e_{\mathrm{t}}})
                                  +{\bm{\mathrm{Y}}}_{\!\!\bm{\mathrm{p}}_{e}}(Z_{e})  -  {\bm{\mathrm{Z}}}_{\bm{\mathrm{p}}_{e}}(Y_{e})
                                  +{\bm{\mathrm{Y}}}_{\!\!\bm{\mathrm{p}}_{A_{\mathrm{t}}}}(Z_{A_{\mathrm{t}}})  -  {\bm{\mathrm{Z}}}_{\bm{\mathrm{p}}_{A_{\mathrm{t}}}}(Y_{A_{\mathrm{t}}})\\
                                & +{\bm{\mathrm{Y}}}_{\!\!\bm{\mathrm{p}}_{A}}(Z_{A})  -  {\bm{\mathrm{Z}}}_{\bm{\mathrm{p}}_{A}}(Y_{A})
                                  +{\bm{\mathrm{Y}}}_{\!\!\bm{\mathrm{p}}_{\phi}}(Z_{\phi})  -  {\bm{\mathrm{Z}}}_{\bm{\mathrm{p}}_{\phi}}(Y_{\phi})\,.
\end{align*}
The vector fields $\mathbb{Y}\in \mathfrak{X}(T^*Q)$ tangent to the primary constraint submanifold $\mathsf{M}_0$ have the following form
\[
\mathbb{Y}=(Y_{e_{\mathrm{t}}}, Y_e, Y_{A_{\mathrm{t}}}, Y_A, Y_{\phi},{\bm{0}},{\bm{0}},{\bm{0}},{\bm{\mathrm{Y}}}_{\!\!\bm{\mathrm{p}}_{A}},{\bm{\mathrm{Y}}}_{\!\!\bm{\mathrm{p}}_{\phi}}),
\]
where
\[
{\bm{\mathrm{Y}}}_{\!\!\bm{\mathrm{p}}_{A}}(\cdot)=\int_\Sigma (\varepsilon_{ijk}Y_e^j\wedge e^k+Y_{ei}\wedge \mathrm{d}\phi+e_i\wedge \mathrm{d}\mathrm{}Y_\phi)\wedge \cdot\,,
\]
and
\[
{\bm{\mathrm{Y}}}_{\!\!\bm{\mathrm{p}}_{\phi}}(\cdot)=\int_\Sigma (-e_i\wedge DY_A^i-Y_e^i\wedge F_i)\cdot\,.
\]

As mentioned before, in the GNH approach, one works directly on the primary constraint submanifold $\mathsf{M}_0$, which, in this case as well, can be immediately identified with the configuration space $Q$. Thus, the vector fields on $\mathsf{M}_0$ have the form
\[
\mathbb{Y}_0=(Y_{e_{\mathrm{t}}}^i, Y_e^i, Y_{A_{\mathrm{t}}}^i, Y_A^i, Y_\phi)\,,
\]
and the pullback of $\Omega$ to $\mathsf{M}_0$ acting on these vector fields is
\begin{align}\label{pullback_omega}
\begin{split}
\omega(\mathbb{Z}_0,\mathbb{Y}_0)=\int_\Sigma\,\Big(&Y_e^i\wedge(\varepsilon_{ijk}e^j\wedge Z_A^k+\mathrm{d}\phi\wedge Z_{\!Ai}-F_iZ_\phi)\\
&\!\!+Y_A^i\wedge(-\varepsilon_{ijk}e^j\wedge Z_e^k+\mathrm{d}\phi\wedge Z_{ei}-Z_\phi De_ i)\\
&\!\!+Y_\phi( F_i\wedge Z_e^i+De_i\wedge Z_A^i)\Big)\,.
\end{split}
\end{align}
We can now solve $\imath_{\mathbb{Z}_0}\omega=\dd H$ by equating the terms proportional to the different components of the vector field $\mathbb{Y}$ in \eqref{dH} and \eqref{pullback_omega} and get the secondary constraints
\begin{align}
&\varepsilon_{ijk}e^j\wedge F^k+\mathrm{d}\phi\wedge F_i=0\,,\label{L1}\\
&\varepsilon_{ijk}e^j\!\wedge\! De^k\!-\mathrm{d}\phi\!\wedge\! De_i=0\,,\label{L2}
\end{align}
and the following equations for the components of the Hamiltonian vector field $\mathbb{Z}_0$
\begin{align}
  & (\varepsilon_{ijk}e^j+\delta_{ik}\mathrm{d}\phi)\wedge(Z_A^k-DA_{\mathrm{t}}^k)=(\varepsilon_{ijk}e_{\mathrm{t}}^j+\delta_{ik}Z_\phi)F^k\,,\label{E1}\\
  & (\varepsilon_{ijk}e^j-\delta_{ik}\mathrm{d}\phi)\wedge(Z_e^k-De_{\mathrm{t}}^k-\varepsilon^k_{\phantom{k}\ell m}e^\ell A_{\mathrm{t}}^m)=(\varepsilon_{ijk}e_{\mathrm{t}}^j-\delta_{ik}Z_\phi)De^k,\label{E2}\\
  & De_i\wedge (Z_A^i-DA_{\mathrm{t}}^i)+F_i\wedge(Z_e^i-De_{\mathrm{t}}^i-\varepsilon^i_{\phantom{i}jk}e^j A_{\mathrm{t}}^k)=0\,.\label{E3}
\end{align}
In terms of the objects
\begin{align}
X_A^i&:=Z_A^i-DA_{\mathrm{t}}^i\,,\\
X_e^i&:=Z_e^i-De_{\mathrm{t}}^i-\varepsilon^i_{\phantom{i}jk}e^jA_{\mathrm{t}}^k\,,
\end{align}
the preceding equations become
\begin{align}
  & (\varepsilon_{ijk}e^j+\delta_{ik}\mathrm{d}\phi)\wedge X_A^k=(\varepsilon_{ijk}e_{\mathrm{t}}^j+\delta_{ik}Z_\phi)F^k\,,\label{EE1}\\
  & (\varepsilon_{ijk}e^j-\delta_{ik}\mathrm{d}\phi)\wedge X_e^k=(\varepsilon_{ijk}e_{\mathrm{t}}^j-\delta_{ik}Z_\phi)De^k,\label{EE2}\\
  & De_i\wedge X_A^i+F_i\wedge X_e^i=0\,.\label{EE3}
\end{align}
As in the standard HK model, there are no conditions on $Z^i_{e_{\mathrm{t}}}$ and $Z^i_{A_{\mathrm{t}}}$. This implies that the evolution of $e_{\mathrm{t}}^i$ and $A_{\mathrm{t}}^i$ is arbitrary. At this point the preceding equations must be solved and the consistency of the dynamics must be checked by showing that the Hamiltonian vector fields are tangent to the constraints. We leave some of the details of this analysis for Appendix \ref{appendix_details_EHK}.

Now, as we did in the case of the standard HK model, we will use the results of the GNH analysis to find a convenient set of new variables to describe the EHK dynamics.

The Hamiltonian formulation obtained above is formulated on the infinite di\-men\-sio\-nal manifold $\mathsf{M}_0$ where the fields $(e_{\mathrm{t}}^i,e^i,A_{\mathrm{t}}^i,A^i,\phi)$ live. The vector fields in $\mathsf{M}_0$ have components $\mathbb{Y}_0=(Y_{e_{\mathrm{t}}}^i, Y_e^i, Y_{A_{\mathrm{t}}}^i, Y_A^i, Y_\phi)$. The pre-symplectic 2-form on $\mathsf{M}_0$ can be written as
\begin{equation}\label{omega_canonical_transf}
\omega=\int_\Sigma \left(\dd A^i\ww\dd\Big(\frac{1}{2}\varepsilon_{ijk}e^j\wedge e^k+e_i\wedge \mathrm{d}\phi\Big)-\dd\phi \ww\dd\left(e_i\wedge F^i\right)\right)\,,
\end{equation}
which on vector fields $\mathbb{Y},\mathbb{Z}$ in $\mathsf{M}_0$ gives \eqref{pullback_omega}. The secondary constraints are
\begin{align*}
&\varepsilon_{ijk}e^j\wedge F^k+\mathrm{d}\phi\wedge F_i=0\,,\\
&\varepsilon_{ijk}e^j\!\wedge\! De^k\!-\mathrm{d}\phi\wedge\! De_i=0\,,
\end{align*}
and the components of the Hamiltonian vector fields
\begin{alignat*}{3}
  & Z_e^i = &&\, De_{\mathrm{t}}^i-\varepsilon^{ijk}A_{\mathrm{t}j}e_k+E^i_{\phantom{i}j}e^j\,, \\
  & Z_A^i = &&\, DA_{\mathrm{t}}^i+W^i_{\phantom{i}j}e^j\,, \\
  & Z_{e_{\mathrm{t}}}^i &&\,\mathsf{arbitrary}\,, \\
  & Z_{A_{\mathrm{t}}}^i &&\,\mathsf{arbitrary}\,,  \\
  & Z_\phi &&\, \mathsf{arbitrary}\,,
\end{alignat*}
with $W_{ij}$ given by \eqref{solution_W} and \eqref{M}, and $E_{ij}$ by \eqref{solution_E} and \eqref{N}.

Now, we will find a Hamiltonian formulation for the present model in terms of canonically conjugate variables in a manner similar to the one described in \cite{BarberoG:2023qih} for the self-dual action. In view of the form of \eqref{omega_canonical_transf} and the discussion presented in Section \ref{sec_Hamiltonian_HK}, it is natural to introduce the following objects
\begin{align}
H_i&:=\frac{1}{2}\varepsilon_{ijk}e^j\wedge e^k+e_i\wedge \mathrm{d}\phi\in \Omega^2(\Sigma)\,,\label{def_H}\\
\pi&:=-e_i\wedge F^i\in\Omega^3(\Sigma)\,,\label{def_pi}
\end{align}
and, with the help of a fiducial volume form $\mathsf{vol}_0$, define
\begin{align}
  \widetilde{H}^i & :=\left(\frac{\cdot\wedge H^i}{\mathsf{vol}_0}\right)\in \mathfrak{X}(\Sigma)\,,\label{def_Htilde} \\
  \widetilde{\pi}&:=\left(\frac{\pi}{\mathsf{vol}_0}\right)\in C^\infty(\Sigma)\,. \label{def_pitilde}
\end{align}

We can view now $\mathsf{M}_0$ as a submanifold $\widetilde{\mathsf{M}}_0$ of $\widetilde{\Gamma}:=C^\infty(\Sigma) \times \mathfrak{X}(\Sigma)^3 \times C^\infty(\Sigma)^3\times \Omega^1(\Sigma)^3\times C^\infty(\Sigma)\times C^\infty(\Sigma)$. To this end, we introduce the injective map (hence, a bijection onto its image)
\[
J:\mathsf{M}_0\rightarrow \widetilde{\Gamma}:(e_{\mathrm{t}}^i,e^i,A_{\mathrm{t}}^i,A^i,\phi)\mapsto(e_{\mathrm{t}}^i,\widetilde{H}^i,A_{\mathrm{t}}^i,A^i,\phi,\widetilde{\pi})\,,
\]
with $\widetilde{H}^i$ and $\widetilde{\pi}$ given by \eqref{def_Htilde} and \eqref{def_pitilde}, which we can think of as a parametrization of the embedding of $\mathsf{M}_0$ into $\widetilde{\Gamma}$ defined by $J$. Our task now is to describe our system in the setting provided by $\widetilde{\Gamma}$.

To begin with, it is straightforward to see that the pre-symplectic form becomes
\begin{equation}\label{omega_tilde}
\widetilde{\omega}=\int_\Sigma\left(\dd A^i\ww\dd \widetilde{H}_i+\dd\phi\ww\dd \widetilde{\pi}\right)\mathsf{vol}_0\,,
\end{equation}
or, in other words, $J^*\widetilde{\omega}=\omega$. This can be seen immediately by plugging \eqref{def_Htilde} and \eqref{def_pitilde} into \eqref{omega_tilde}. This is just a pre-symplectic form because neither $A_{\mathrm{t}}^i$ nor $e_{\mathrm{t}}^i$ appear in it. However, it has the standard ``canonical'' form in the pairs of conjugate variables $(A^i, \widetilde{H}^i)$ and $(\phi,\widetilde{\pi})$.

Next, we will write the constraints in terms of the new variables introduced above. There are, in fact, two types of constraints that we have to consider: the one given by \eqref{def_pi}, and the secondary constraints \eqref{L1}-\eqref{L2}. Notice that, loosely speaking, whereas a counting of the number of independent field components allows us to interpret the pair $(\widetilde{H}^i,\phi)$ as a simple redefinition of $(e^i,\phi)$ on account of \eqref{def_Htilde}, the expression involving $\pi$ must be kept as the condition that defines $J(\mathsf{M}_0)$ as a proper submanifold of $\widetilde{\Gamma}$.

The condition \eqref{def_pi} can be written as
\begin{equation}\label{constr_pi}
2\widetilde{\pi}-\big((\mathrm{det} h)^2-({\widetilde{H}_\ell}\iprod\mathrm{d}\phi)({\widetilde{H}^\ell}\iprod\mathrm{d}\phi)\big)^{-1/2}\varepsilon^{ijk}{\widetilde{H}_i}\iprod{\widetilde{H}_j}\iprod F_k=0\,,
\end{equation}
with
\[
(\mathrm{det}h)^2=-\frac{1}{3!}\varepsilon^{ijk}{\widetilde{H}_i}\iprod{\widetilde{H}_j}\iprod{\widetilde{H}_k}\iprod\mathsf{vol}_0\,.
\]
This corresponds to the scalar constraint appearing in \cite{BarberoG:1997nrd}. As we can see, the second term in \eqref{constr_pi} is proportional to the Hamiltonian constraint appearing in the Ashtekar formulation for Euclidean gravity. As a consequence, we expect that the field dynamics will contain contributions with a form similar to the one of Euclidean GR and some extra terms. This will be clearly seen when we look at the Hamiltonian vector fields.

We now write the constraint \eqref{L1} in terms of the new variables. This is simply
\begin{equation}\label{constr_vec}
  {\widetilde{H}_i}\iprod F^i-\widetilde{\pi}\mathrm{d}\phi=0\,.
\end{equation}
It has the form expected of a vector constraint with two terms generating diffeomorphisms on the $(A^i,\widetilde{H}^i)$ and $(\phi,\widetilde{\pi})$ variables, respectively. Finally, a computation essentially similar to the one leading to \eqref{const_st_Ei1}, turns \eqref{L2} into
\begin{equation}\label{constr_Gauss}
\mathrm{div}_0\widetilde{H}^i+\varepsilon^{ijk}{\widetilde{H}_k}\iprod A_j=0\,,
\end{equation}
which is the familiar Gauss law that generates internal $SO(3)$ gauge transformations.

To understand the dynamics of the model in terms of the variables introduced in this section we have to write $Z_A^i$ and $Z_{\widetilde{H}}^i$ in terms of them. In the case of $Z_A^i$ the computation simply amounts to rewriting it in terms of the $h_i$ and $^h \mathbb{F}_{ij}$ defined in Appendix \ref{appendix_identities} and also in terms of
\begin{equation}\label{lapseshift}
\widehat{e}_{\mathrm{t}i}:=\frac{e_{\mathrm{t}i}+Z_\phi\alpha_i-\varepsilon_{ijk}e_{\mathrm{t}}^j\alpha^k}{\sqrt{1+\alpha^2}}\,,\quad \widehat{\alpha}_{\mathrm{t}}:=\frac{Z_\phi-e_{\mathrm{t}}\cdot\alpha}{\sqrt{1+\alpha^2}}\,.
\end{equation}
The result is
\begin{align}
Z_A^i&=DA_{\mathrm{t}}^i-\widehat{\alpha}_{\mathrm{t}} {^h \mathbb{F}^{ij}h_j}+{^h \mathbb{F}}^i_{\phantom{i}j}h_k\varepsilon^{kj\ell}  \widehat{e}_{\mathrm{t}\ell}+\widehat{\alpha}_{\mathrm{t}} {^h\mathbb{F}}\Big(-\varepsilon^{ijk}\alpha_j+\frac{1}{2}(1-\alpha^2)\delta^{ik}-\alpha^i\alpha^k\Big)h_k\label{ZAh}\\
     &=\pounds_{\xi}A^i+D(A_{\mathrm{t}}^i-{\xi}\iprod A^i)-\widehat{\alpha}_{\mathrm{t}} {^h \mathbb{F}^{ij}h_j}+\widehat{\alpha}_{\mathrm{t}} {^h\mathbb{F}}\Big(-\varepsilon^{ijk}\alpha_j+\frac{1}{2}(1-\alpha^2)\delta^{ik}-\alpha^i\alpha^k\Big)h_k \,,\nonumber
\end{align}
with $\xi\in\mathfrak{X}(\Sigma)$ the unique vector field satisfying ${\xi}\iprod h^i=\widehat{e}_{\mathrm{t}}^{\,\,i}$ or, equivalently, $\xi\iprod e^i=e_{\mathrm{t}}^i$ as shown in Appendix \ref{uniqueness_xi}. In the previous expression $h_i$ should be understood as written in terms of $\widetilde{H}^i$. As we can see the evolution of $A^i$ given by $Z_A^i$ consists of two contributions: one proportional to $\widehat{\alpha}_{\mathrm{t}} {^h\mathbb{F}}$ and the rest, that corresponds \textit{exactly} to the dynamics of the Ashtekar connection in full Euclidean GR as defined by the lapse $\widehat{\alpha}_{\mathrm{t}}$ and the shift $\xi$.

Let us consider now $Z_{\widetilde{H}}^i$. From the fact that $H_i=\widetilde{H}_i\iprod \mathsf{vol}_0$ we deduce that $Z_{Hi}$ and $Z_{\widehat{H}i}$ are related by
\[
Z_{Hi}=Z_{\widetilde{H}i}\iprod \mathsf{vol}_0\,,
\]
from which $Z_{\widetilde{H}i}$ can be written in terms of $Z_{Hi}$. The computation of $Z_{Hi}$ can be sketched as follows. From the definition of $H_i$ [equation \eqref{def_H}] we find
\begin{equation}\label{ZHi}
Z_{Hi}=\varepsilon_{ijk}e^j\wedge Z_e^k+Z_{ei}\wedge \mathrm{d}\phi+e_i\wedge \mathrm{d}Z_\phi\,.
\end{equation}
Now, from \eqref{E2} we have
\begin{equation*}
\varepsilon_{ijk}e^j\wedge Z_e^k+Z_{ei}\wedge \mathrm{d}\phi=(\varepsilon_{ijk}e^j-\delta_{ik}\mathrm{d}\phi)\wedge(De_{\mathrm{t}}^k+\varepsilon^{k\ell m}e_\ell A_{\mathrm{t}m})+(\varepsilon_{ijk}e_{\mathrm{t}}^j-\delta_{ik}Z_\phi)De^k\,,
\end{equation*}
which, plugged into \eqref{ZHi} gives
\[
Z_{Hi}=D(\varepsilon_{ijk}e_{\mathrm{t}}^je^k-Z_\phi e_i+e_{\mathrm{t}i}) \mathrm{d}\phi -\varepsilon_{ijk}A_{\mathrm{t}}^jH^k\,.
\]
Writing now $e_i$ in terms of $h_i$, which can be expressed in terms of $H_i$ as
\[
h_i=-\frac{1}{2\mathrm{det}h}\varepsilon_{ijk}\widetilde{H}^j\iprod \widetilde{H}^k\iprod\mathsf{vol}_0\,,
\]
we arrive at
\begin{equation}\label{ZHi2}
Z_{Hi}=D(\varepsilon_{ijk}\widehat{e}_{\mathrm{t}}^{\,\,j}h^k)-D(\widehat{\alpha}_{\mathrm{t}}h_i)-\varepsilon_{ijk}A_{\mathrm{t}}^jH^k=\pounds_{\xi}H_i-\varepsilon_{ijk}(A_{\mathrm{t}}^j-\xi\iprod\! A^j)H^k-D(\widehat{\alpha}_{\mathrm{t}}h_i) \,.
\end{equation}
This is, \textit{precisely}, (see \cite{BarberoG:2023qih}) the evolution of $H_i$ in Euclidean GR defined by the lapse $\widehat{\alpha}_\mathrm{t}$ and the shift $\widehat{e}_{\mathrm{t}}^{\,\,i}$ given by \eqref{lapseshift}.

As far as the rest of the components of the Hamiltonian vector fields are concerned, we still have that $Z_{e_{\mathrm{t}}}$, $Z_{A_{\mathrm{t}}}$ and $Z_\phi$ are arbitrary and, from \eqref{def_pi} and \eqref{def_pitilde}, we can obtain $Z_{\widetilde{\pi}}$ in terms of $Z_A^i$ and $Z_{\widetilde{H}}^i$.

To interpret and understand the dynamics of the Extended HK model discussed above, one could take different approaches. One possibility would be to impose the gauge fixing condition $\phi=0$ (which would lead to the consistency requirement $Z_\phi=0$). By doing this, we would have $\alpha_i=0$, and recover the standard HK Hamiltonian formulation. Note, however, that this procedure is somehow unsatisfactory because the absence of $\phi$ would preclude us from building non-degenerate 4-dimensional metrics, which, as commented in \cite{BarberoG:1997nrd}, is one of the reasons to include the $\phi$ field as a time variable.

Another interesting possibility is to take advantage of the arbitrariness of $Z_\phi$ and choose $Z_\phi=e_{\mathrm{t}}\cdot\alpha$. This would lead to
\begin{equation*}
\widehat{e}_{\mathrm{t}i}=\frac{e_{\mathrm{t}i}+(e_{\mathrm{t}}\cdot\alpha)\alpha_i-\varepsilon_{ijk}e_{\mathrm{t}}^j\alpha^k}{\sqrt{1+\alpha^2}}\,,\quad \widehat{\alpha}_{\mathrm{t}}=0\,.
\end{equation*}
If we now plug $\widehat{\alpha}_{\mathrm{t}}=0$ into the expressions for $Z_A^i$ and $Z_{H}^i$ we get
\begin{align*}
  & Z_A^i=\pounds_\xi A^i+D(A_{\mathrm{t}}^i-\xi\iprod A^i)\,, \\
  & Z_H^i=\pounds_\xi H^i-\varepsilon^{ijk}(A_{\mathrm{t}j}-\xi\iprod A_j)H_k\,,
\end{align*}
with the vector field $\xi\in\mathfrak{X}(\Sigma)$ satisfying $\xi\iprod e^i=e_{\mathrm{t}}^i$ and, crucially, also $\xi\iprod h^i=\widehat{e}_{\mathrm{t}}^{\,\,i}$. As we can see we recover the usual dynamics of the standard HK model for the variables $A^i$ and $H^i$ (or, equivalently, $\widehat{H}^i$), although it must be remembered that the constraints are \eqref{constr_vec} and \eqref{constr_Gauss}. Several comments are in order now:

Strictly speaking, the choice $Z_\phi=e_{\mathrm{t}}\cdot\alpha$ is not a gauge fixing but, rather, a particular way of specifying the, otherwise arbitrary, dynamics of $\phi$. Gauge fixing conditions restrict the dynamics of singular systems to submanifolds of the phase space with the hope of avoiding redundancies in the specification of the physically relevant configurations. Very often (partial gauge fixings), they ``reduce the dimensionality'' of the gauge orbits. Of course, once a good gauge fixing condition is introduced, the arbitrary components of the Hamiltonian vector fields must be adjusted in such a way that these are still tangent to the submanifold in phase space determined by the gauge fixing. Here we have followed a slightly different approach: choosing directly some of the arbitrary components of the Hamiltonian vector fields as a way of fixing the corresponding fields at any time once their initial values are given. In fact, if we notice that
\[
\pounds_\xi\phi=\xi\iprod\mathrm{d}\phi=(\xi\iprod e^i)\alpha_i =e_{\mathrm{t}}\cdot\alpha=Z_\phi\,,
\]
we see that the dynamics of the scalar field $\phi$ is given by $\dot{\phi}=\pounds_\xi \phi$, i.e. the field $\phi$ is Lie-dragged along the direction defined by the vector field $\xi$. As a consequence, if $\phi$ is nowhere zero at some ``moment'' $t_0$, it will remain nowhere zero for every $t$. This avoids the problems associated with the gauge fixing $\phi=0$ mentioned above.

It is natural to wonder what happens if we insert the choice $Z_\phi=e_{\mathrm{t}}\cdot\alpha$ into \eqref{solution_W} and \eqref{solution_E} (i.e. before introducing any new variables). The result, as expected, is that one gets the dynamics of the standard HK model \eqref{Z_st} with the constraints \eqref{L1} and \eqref{L2}.

%
%
\section{Conclusions and comments}\label{sec_conclusions}

As stated in the introduction, the main goal of the paper is to study in detail the dynamics of the Extended Husain-Kucha\v{r} model originally presented in \cite{BarberoG:1997nrd} by disentangling the form of its Hamiltonian vector fields. One of our motivations is to see if a gauge fixing condition such as $\phi=0$ could be avoided. The reason for this is similar to the one invoked in the ADM formalism: even if the lapse and shift are arbitrary variables, the lapse must always be taken to be different from zero in order to build four-dimensional non-degenerate metrics from the dynamical three metrics. Here we have found a way to avoid similar problems in the context of the HK model by relying on insights gained in the derivation of the Ashtekar formulation from the self-dual action without using the time gauge \cite{BarberoG:2023qih}. In particular, the introduction of a redefined lapse and shift according to \eqref{lapseshift} leads in a neat way to the interpretation of the Extended model as the standard one.

We would like to make some comments. The first one is related to the way the Dirac approach to constrained systems is often used. In most of the literature (and certainly in Dirac's book \cite{Dirac}), the main goal is quantization, for which the constraints play a central role. Although the dynamics can be explicitly derived by computing the Poisson brackets of the relevant Hamiltonians (incorporating the primary constraints), it is seldom discussed in much detail, in particular in the context of the Ashtekar formulation that is usually derived from actions (self-dual \cite{Samuel:1987td,Jacobson:1987yw} or Holst \cite{Holst:1995pc}) by using the time gauge.

One of the reasons to proceed in the way followed by Dirac is to take advantage of the canonical symplectic form because it can be easily quantized. In this regard, the GNH approach that we have used here is less suited because, generically, the pullback of the canonical symplectic form to the primary constraint submanifold is degenerate and does not have the ``canonical form''. This notwithstanding, in the examples discussed here (and also in \cite{BarberoG:2023qih}), this is not much of a problem because it is either possible to find variables in terms of which the final phase space carries a symplectic structure with a canonical form (self-dual Euclidean gravity or the standard HK model), or it is possible to extend it in a simple way in order to achieve this goal (the EHK model). From a conceptual perspective, it is important to remember that some variables with arbitrary dynamics (such as $A_{\mathrm{t}}^i$ and $e_{\mathrm{t}}^i$ in the models discussed in the paper) can consistently be interpreted as ``external fields'' that can be arbitrarily specified. By proceeding in this way, they can be removed from the definition of the configuration space so that their absence in the pullback of the symplectic to the primary constraint submanifold does not render it degenerate. On the other hand, as the EHK model shows, it is possible to have dynamical variables such as $\phi$, with arbitrary dynamics, explicitly appearing in the symplectic structure.

We have insisted on the fact that, by looking at the form of the pullback $\omega$ of the canonical symplectic form $\Omega$ onto the primary constraint submanifold $\mathsf{M}_0$, it is possible to guess new variables that simplify both the form of $\omega$ and the constraints. With hindsight, it can be argued that these variables can be directly guessed by looking at the fiber derivative and, hence, used in the framework of Dirac's approach. This is certainly true and shows that the Dirac and GNH points of view are actually quite close. Notice, however, that despite the appearances, the actual implementation of the GNH method is often much shorter than the one of Dirac's algorithm (the sad reason is that the necessary consistency requirements are seldom checked, and the final form of the total Hamiltonian is often not given).

One of the interesting aspects of the approach that we followed in \cite{BarberoG:2023qih} and in this paper is that it provides a systematic way to arrive at the Ashtekar formulation with minimal guesswork (in particular, gauge fixing can be avoided). In our opinion, it can lead to new Hamiltonian formulations for background independent models and gravitational theories which may be interesting both from the classical and quantum perspectives.

The most immediate application of the results and methods discussed here is the introduction of a time variable in full GR by copying the mechanism used to arrive at \eqref{EHK_action}. Notice that by performing integration by parts in its second term (see the discussion in \cite{BarberoG:1997nrd}), this term can be written as the product of a Lagrange multiplier times an object that is zero as a consequence of the field equations derived from the action given by just the first term in \eqref{EHK_action} [i.e., the action for the standard HK model \eqref{HK_action}]. By following a similar approach, it may be possible to introduce time variables in the Euclidean self-dual or Holst actions. We plan to look at this problem in the near future.

%
%
\section*{Acknowledgments}

This work has been supported by the Spanish Ministerio de Ciencia Innovaci\'on y Universidades-Agencia Estatal de Investigaci\'on grant AEI/PID2020-116567GB-C22. E.J.S. Villase\~nor is supported by the Madrid Government (Comunidad de Madrid-Spain) under the Multiannual Agreement with UC3M in the line of Excellence of University Professors (EPUC3M23), and in the context of the
V PRICIT (Regional Programme of Research and Technological Innovation).

\begin{appendices}

%
%

\section{Notation and definitions}\label{appendix_notation}

We summarize here the conventions and notations used in the paper. Many of the objects that we employ are differential forms, so we will avoid the use of space and spacetime indices. However, the geometric objects that we use carry internal indices that we will keep explicit. These are written with Latin letters and take values $1,2,3$. They will be raised and lowered with the Euclidean diagonal metric $\mathrm{Diag}(+1,+1,+1)$. We will make extensive use of the 3-dimensional Levi-Civita totally antisymmetric symbol $\varepsilon_{ijk}$.

The exterior differential in three and four dimensions will be respectively denoted by $\mathrm{d}$ and $\bd$, and the exterior differential in phase space as $\dd$. The exterior product of forms in space or spacetime will be denoted by $\wedge$, whereas in phase space, it will be denoted by $\w$.

Interior products of vector fields and forms in space or spacetime will be denoted by $\iprod$ (for instance, $\xi\iprod \beta$ for a vector field $\xi$ and a 1-form $\beta$ on the spatial manifold $\Sigma$). Interior products in phase space will be denoted by $\imath$ (for instance, $\imath_{\mathbb{Y}}\Omega$ for a vector field $\mathbb{Y}$ and a 2-form $\Omega$).

If $\bm\alpha$ and $\mathsf{vol}$ are respectively a top form and a volume form in a manifold $\mathcal{M}$, we denote as
\[
\left(\frac{\bm\alpha}{\mathsf{vol}}\right)
\]
the unique scalar field such that
\[
\bm\alpha=\left(\frac{\bm\alpha}{\mathsf{vol}}\right)\mathsf{vol}\,.
\]

Throughout the paper, we make use of a fiducial volume form $\mathsf{vol}_0$ (remember that we restrict ourselves to orientable spatial manifolds $\Sigma$). We will always choose it in such a way that $\mathrm{det}e:=\left(\frac{\mathsf{vol}_e}{\mathsf{vol}_0}\right)>0$ and $\mathrm{det}h:=\left(\frac{\mathsf{vol}_h}{\mathsf{vol}_0}\right)>0$.

\section{Some useful identities}\label{appendix_identities}

In this appendix, we define a number of auxiliary geometric objects necessary to understand the dynamics of the EHK model and list some useful properties of them. We start by reminding the readers that we expand $\mathrm{d}\phi=\alpha^ie_i$ and write $\alpha^2:=\alpha_i\alpha^i$.

We define the 2-forms $H_i$ ($i=1,2,3$) as
\begin{equation*}
  H_i:=\frac{1}{2}\varepsilon_{ijk}e^j\wedge e^k+e_i\wedge \mathrm{d}\phi\,.
\end{equation*}
It is useful to introduce auxiliary triads $h_i\in\Omega^2(\Sigma)$ such that $H_i=\frac{1}{2}\varepsilon_{ijk}h^j\wedge h^k$. In terms of $e_i$ and $\mathrm{d}\phi$ (equivalently, $\alpha_i$) they are given by
\begin{equation}\label{eq_app_1}
  h^i=\frac{1}{\sqrt{1+\alpha^2}}(\delta^i_{\phantom{i}k}+\alpha^i\alpha_k+\varepsilon^{ij}_{\phantom{ij}k}\alpha_j)e^k\,.
\end{equation}
Reciprocally, we have
\begin{equation*}
e_i=\frac{1}{\sqrt{1+\alpha^2}}(\delta_i^{\phantom{i}j}+\varepsilon_i^{\phantom{i}jk}\alpha_k)h_j\,.
\end{equation*}
With the help of the $h_i$ we define the 3-form
\begin{equation*}
 \frac{1}{3!}\varepsilon_{ijk}h^i\wedge h^j\wedge h^k,
\end{equation*}
that we will denote as $\mathsf{vol}_h$ because it is a volume form if and only if $\mathsf{vol}_e$ (defined in the main text) is also a volume form. Indeed, a straightforward computation gives
\begin{equation*}
\mathsf{vol}_h=\sqrt{1+\alpha^2}\,\mathsf{vol}_e\,.
\end{equation*}
The vector fields $\widetilde{H}_i$ used throughout the paper are defined by duality in terms of $H_i$ (and, also, $h_i$) according to
\begin{equation*}
\widetilde{H}_i:=\left(\frac{\cdot\wedge H_i}{\mathsf{vol}_0}\right)\,.
\end{equation*}
Some useful properties of these objects are:
\begin{align*}
  & {\widetilde{H}_i}\iprod\mathsf{vol}_0=H_i\,, \\
  & {\widetilde{H}_i}\iprod h^j=(\mathrm{det}h)\delta_i^{\phantom{i}j}\,, \\
  & (\mathrm{det}h)^2=-\frac{1}{3!}\varepsilon_{ijk}\widetilde{H}^i\iprod \widetilde{H}^j\iprod \widetilde{H}^k \iprod\mathsf{vol}_0\,,\\
  & {\widetilde{H}_i}\iprod e^j=\frac{\mathrm{det}h}{\sqrt{1+\alpha^2}}(\delta_i^{\phantom{i}j}-\varepsilon_i^{\phantom{i}jk}\alpha_k)\,, \\
  & \pounds_{\widetilde{H}_i}\phi={\widetilde{H}_i}\iprod\mathrm{d}\phi=\frac{\mathrm{det}h}{\sqrt{1+\alpha^2}}\alpha_i\,, \\
  & \sqrt{1+\alpha^2}=\frac{\mathrm{det}h}{\sqrt{(\mathrm{det}h)^2-(\pounds_{\widetilde{H}_i}\phi)(\pounds_{\widetilde{H}^i}\phi)}}\,.
\end{align*}

In analogy with the definition of $\mathbb{F}_{ij}$ and $\mathbb{B}_{ij}$ we introduce
\begin{equation*}
  ^h\mathbb{F}_{ij}:=\left(\frac{F_i\wedge h_j}{\mathsf{vol}_h}\right)\,,\quad ^h\mathbb{B}_{ij}:=\left(\frac{Dh_i\wedge h_j}{\mathsf{vol}_h}\right)\,.
\end{equation*}
The relation between $\mathbb{F}_{ij}$ and $^h\mathbb{F}_{ij}$ can be found by using \eqref{eq_app_1}:
\begin{equation*}
  \mathbb{F}_{ij}={^h\mathbb{F}}_{ij}+\varepsilon_{jk\ell} {^h\mathbb{F}}_i^{\phantom{i}k}\alpha^\ell\qquad\longrightarrow\qquad
^h\mathbb{F}:={^h\mathbb{F}}^k_{\phantom{k}k}=\frac{\mathbb{F}}{1+\alpha^2}\,.
\end{equation*}
In terms of $^h\mathbb{F}_{ij}$ and $^h \mathbb{B}_{ij}$ the secondary constraints of the EHK model can be written in the form
\begin{align*}
  & \alpha_i\, ^h\mathbb{F}-\varepsilon_{ijk} {^h\mathbb{F}^{jk}}=0\,, \\
  & \varepsilon_{ijk} {^h\mathbb{B}^{jk}}=0\,.
\end{align*}

\section{Some mathematical results}\label{appendix_math_details}

\subsection{Coframes and volume forms}\label{coframes_volume}

Here we prove the following elementary lemma:

\medskip

\begin{lemma}\label{lemma volume form}
    Let $\mathcal{M}$ be a compact and parallelizable 4-dimensional manifold, then $(\ee^0,\ee^i)$ ($i=1,2,3$) is a coframe if and only if
\[
\mathsf{vol}_0:=\frac{1}{3!}\varepsilon_{ijk}\ee^0\wedge \ee^i\wedge \ee^j\wedge \ee^k
\]
is a volume form.
\end{lemma}
\begin{proof}

\noindent $\boxed{\Rightarrow}$ If $(\ee^0,\ee^i)$ ($i=1,2,3$) is a coframe then, at every point $p\in \mathcal{M}$ the covectors $\ee^0(p)\,, \ee^i(p)$ ($i=1,2,3$) are linearly independent. Also, a frame $(E_0,E_i)$ ($i=1,2,3$) exists in such a way that, at every $p\in \mathcal{M}$, $(\ee^0(p)\,, \ee^i(p))$ ($i=1,2,3$) is the dual basis of $(E_0(p),E_i(p))$ ($i=1,2,3$) satisfying
\[
{E_0}\iprod\ee^0=1\,, {E_0}\iprod\ee^i=0\,, {E_i}\iprod\ee^0=0\,, {E_i}\iprod\ee^j=\delta_i^{\phantom{i}j}\,, (i,j=1,2,3)\,.
\]
Now a straightforward computation gives
\[
{E_0}\iprod{E_1}\iprod{E_2}\iprod{E_3}\iprod\mathsf{vol}_0=1\,,
\]
which is a constant, non-vanishing function on $\mathcal{M}$, hence, $\mathsf{vol}_0$ is a volume form as it vanishes nowhere on $\mathcal{M}$.

\medskip

\noindent $\boxed{\Leftarrow}$ If at $p\in \mathcal{M}$ the covectors $\ee^0(p)\,, \ee^i(p)$ ($i=1,2,3$) are linearly dependent then $\varepsilon_{ijk}\ee^0(p)\wedge \ee^i(p)\wedge \ee^j(p)\wedge \ee^k(p)=0$ and, hence, $\mathsf{vol}_0$ is not a volume form.
\end{proof}

\subsection{\texorpdfstring{A result on internal products involving $\tilde{u}_0$}{A result on internal products involving \textbackslash tilde\{u\}\_0}}\label{tildeu0}

\medskip

\begin{lemma}\label{int_prods} Let $(\bm{\xi}^k)\in\Omega^2(\mathcal{M})^3$, then $\varepsilon_{ijk}{\ee}^j\wedge {\bm{\xi}}^k=0$ implies $\tilde{u}_0\iprod {\bm{\xi}}^k=0$.\end{lemma}

\begin{proof}

\noindent Since ${\tilde{u}_0}\iprod\bm{\mathrm{e}}^i=0$, we see that $\varepsilon_{ijk}\bm{\mathrm{e}}^j\wedge{\bm{\xi}}^k=0$ implies $\varepsilon_{ijk}\bm{\mathrm{e}}^j\wedge(\tilde{u}_0\iprod{\bm{\xi}}^k)=0$. We now show that this condition leads to ${\tilde{u}_0}\iprod{\bm{\xi}}^k=0$. Let us expand ${\tilde{u}_0}\iprod{\bm{\xi}}^k=\alpha^k_{\phantom{k}\ell}\ee^\ell+\alpha {\bm{\mathrm{d}}}t$. As ${\tilde{u}_0}\iprod{\tilde{u}_0}\iprod{\bm{\xi}}^k=0$ we see that $\alpha=0$ and, hence, we can simply write ${\tilde{u}_0}\iprod{\bm{\xi}}^k=\alpha^k_{\phantom{k}\ell}\ee^\ell$. Now
    \[
    0=\varepsilon_{ijk}{\bm{\mathrm{e}}}^j\wedge {\tilde{u}_0}\iprod{\bm{\xi}}^k=\varepsilon_{ijk}{\bm{\mathrm{e}}}^j\wedge\alpha^k_{\phantom{k}\ell}\ee^\ell\,.
    \]
   Left-wedge-multiplying the previous expression by ${\bm{\mathrm{d}}}t\wedge \bm{\mathrm{e}}^m$ we get
    \[
    0=\alpha^k_{\phantom{k}\ell}\varepsilon_{ijk}{\bm{\mathrm{d}}}t\wedge \bm{\mathrm{e}}^m\wedge\bm{\mathrm{e}}^j\wedge  \bm{\mathrm{e}}^\ell=\alpha^k_{\phantom{k}\ell} \varepsilon_{ijk}\varepsilon^{mj\ell}\mathsf{vol}_t =\delta^m_{\phantom{m}i}\alpha^p_{\phantom{p}p}-\alpha^m_{\phantom{m}i}\,.
    \]
    The trace of the previous equation gives $\alpha^p_{\phantom{p}p}=0$ and, hence, $\alpha^i_{\phantom{i}j}=0$, i.e. ${\tilde{u}_0}\iprod{\bm{\xi}}^k=0$.
\end{proof}

\subsection{Pullbacks of coframes}\label{coframes_pullbacks}

To state and prove the following lemma, we remind the reader that for a compact and orientable 3-dimensional manifold $\Sigma$, we introduce the four manifold $\mathcal{M}=\mathbb{R}\times\Sigma$ and the sheets of the canonical foliation of $\mathcal{M}$ given by $\Sigma_\tau:=\{\tau\}\times\Sigma$ ($\tau\in\mathbb{R}$). We have a real function $t\in C^\infty(\mathcal{M})$ and a vector field $\partial_t\in\mathfrak{X}(\mathcal{M})$. Also, for each $\tau\in\mathbb{R}$ we have the embedding $\jmath_\tau:\Sigma\rightarrow \mathcal{M}:p\mapsto (\tau,p)$.

\begin{lemma}\label{pullbacks}

Let us assume that $(\bd t,\ee^i)$ ($i=1,2,3$) is a coframe on $\mathcal{M}$. Then, for all $\tau\in\mathbb{R}$, the three 1-forms $\jmath^*_\tau \ee^i$ ($i=1,2,3$) define a coframe on $\Sigma$.
\end{lemma}

\begin{proof}

Let us assume that, for all $\tau\in\mathbb{R}$, we have $\lambda_i\jmath_\tau^* \ee^i=0$ with $\lambda_i\in C^\infty(\Sigma)$. Let us also introduce the function
\[
{\bm{\lambda}}_i:\mathcal{M}\rightarrow \mathbb{R}:(\tau,p)\mapsto \lambda_i(p)\,.
\]
Now, $\lambda_i\jmath_\tau^*\ee^i(p)=0$ for every $p\in \Sigma$ is equivalent to $\jmath_\tau^*({\bm{\lambda}}_i\ee^i)(p)=0$ for every $p\in \Sigma$. This is also equivalent to $V\iprod({\bm{\lambda}}_i\ee^i)=0$ for every vector field $V\in\mathfrak{X}(\mathcal{M})$ such that $V\iprod\bd t=\pounds_V t=0$ (i.e. vector fields tangent to every $\Sigma_\tau$). In particular, this must be true for $E_j$, the element of the dual basis of $(\bd t,\ee^i)$ satisfying ${E_i}\iprod \ee^j=\delta_i^{\phantom{i}j}\,,{E_i}\iprod\bd t=0$. Hence, for $i=1,2,3$
\[
E_i \iprod({\bm{\lambda}}_j\ee^j)={\bm{\lambda}}_i=0\,,
\]
which implies $\lambda_i=0$, so that the $\jmath_\tau^*\ee^i(p)$ are linearly independent for all $p\in\Sigma_\tau$.
\end{proof}

\subsection{A useful result to interpret the field equations \texorpdfstring{\eqref{field_equations_EHK}}{2.10}}\label{two-forms}

Let us assume, as mentioned above, that the 1-forms $(\bd t,{\bm{\mathrm{e}}}^i)$ ($i=1,2,3$), are linearly independent everywhere on $\mathcal{M}$. This implies, in particular, that every 2-form on $\mathcal{M}$ can be expanded in terms of $\bd t\wedge {\bm{\mathrm{e}}}^i$ and ${\bm{\mathrm{e}}}^i\wedge {\bm{\mathrm{e}}}^j$. We now prove the following lemma (which generalizes the result proved in \cite{BarberoG:1997nrd}):

\begin{lemma}\label{useful_result}

Let ${\bm{\xi}}^k, {\bm{\zeta}}^k\in\Omega^2(\mathcal{M})$ ($k=1,2,3$). Then, if both $\varepsilon_{ijk}{\bm{\mathrm{e}}}^j\wedge{\bm{\xi}}^k-\bd \bm\phi\wedge{\bm{\xi}}_i=0$ and $\varepsilon_{ijk}{\bm{\mathrm{e}}}^j\wedge{\bm{\zeta}}^k+\bd \bm\phi\wedge{\bm{\zeta}}_i=0$ hold for $i=1,2,3$, and $\bm\phi\in C^\infty(\Sigma)$, we also have ${\bm{\xi}}_k\wedge{\bm{\zeta}}^k=0$.
\end{lemma}

\begin{proof}

Let us expand
\begin{align*}
  {\bm{\xi}}^k & =\xi^k_{\phantom{k}\ell}\bd t\wedge \ee^\ell+X^k_{\phantom{k}\ell}\,\varepsilon^{\ell pq}\ee_p\wedge\ee_q\,, \\
  {\bm{\zeta}}^k & =\zeta^k_{\phantom{k}\ell}\bd t\wedge \ee^\ell+Z^k_{\phantom{k}\ell}\,\varepsilon^{\ell pq}\ee_p\wedge\ee_q\,, \\
  \bd\bm\phi & =\alpha_0\bd t+\alpha_i\ee^i,
\end{align*}
with $\xi^k_{\phantom{k}\ell}\,,\zeta^k_{\phantom{k}\ell}\,,\alpha_0\,,X^k_{\phantom{k}\ell}\,,Z^k_{\phantom{k}\ell}\,,\alpha_i\in C^\infty(\mathcal{M})$. Thus, we can write the first condition as
\begin{align*}
    0 = & \,\varepsilon_{ijk}\xi^k_{\phantom{k}\ell}\ee^j\wedge\bd t\wedge\ee^\ell + \varepsilon_{ijk}X^k_{\phantom{k}\ell}\,\varepsilon^{\ell pq}\ee^j\wedge\ee_p\wedge\ee_q - \alpha_0\,\varepsilon^{\ell pq}X_{i\ell}\bd t\wedge\ee_p\wedge\ee_q - \alpha_j\xi_{i\ell}\ee^j\wedge\bd t\wedge\ee^\ell\\
    & - \alpha_j X_{i\ell}\,\varepsilon^{\ell pq}\ee^j\wedge\ee_p\wedge\ee_q.
\end{align*}
Left-wedge-multiplying the previous expression by $\bd t$ leads to
\begin{equation}\label{appendix_lemma4_1}
\varepsilon_{ijk}X^{jk} + X_{ij}\alpha^j = 0,
\end{equation}
and left-wedge-multiplying it by $\ee_m$ leads to
\begin{equation}\label{appendix_lemma4_2}
    \xi_{ij} - \delta_{ij}\xi^k_{\phantom{k}k} + \varepsilon_{ik\ell}\alpha^k\xi_j^{\phantom{\ell}\ell} = 2\alpha_0 X_{ji}.
\end{equation}
In a completely analogous way, we get the conditions
\begin{align}
    &\varepsilon_{ijk}Z^{jk} - Z_{ij}\alpha^j = 0,
    \label{appendix_lemma4_3}\\
    &\zeta_{ij} - \delta_{ij}\zeta^k_{\phantom{k}k} - \varepsilon_{ik\ell}\alpha^k\zeta_j^{\phantom{j}\ell} = -2\alpha_0Z_{ji}.
    \label{appendix_lemma4_4}
\end{align}
If we compare now \eqref{appendix_lemma4_1} and \eqref{appendix_lemma4_3} with the constraints \eqref{LL1} and \eqref{LL2} we see that they are formally equal if we replace $X_{ji}\mapsto\mathbb{B}_{ji}$ and $Z_{ji}\mapsto\mathbb{F}_{ji}$. Furthermore, equations \eqref{appendix_lemma4_2} and \eqref{appendix_lemma4_4} become \eqref{equ_W} and \eqref{equ_E} under the replacements $e_{\mathrm{t}i}\mapsto0$, $Z_\phi\mapsto2\alpha_0$, $\xi_{ij}\mapsto E_{ij}$, and $\zeta_{ij}\mapsto W_{ij}$. Finally,
\begin{equation*}
  {\bm{\xi}}_k\wedge{\bm{\zeta}}^k=(\xi_{ij}Z^{ij} + \zeta_{ij}X^{ij})\mathsf{vol}_t,
\end{equation*}
hence, we need to prove that $\xi_{ij}Z^{ij} + \zeta_{ij}X^{ij} = 0$. Under the replacements introduced above, we see that this last condition is equivalent to proving \eqref{tercera_ec}. As a consequence, the same computation showing that \eqref{tercera_ec} holds leads to the result that we are set to prove here. As mentioned in Section \ref{sec_vector_fields}, this is straightforward, although it requires some patience if performed by hand.

\end{proof}

\subsection{\texorpdfstring{Uniqueness of $\xi$}{Uniqueness of \textbackslash xi}}\label{uniqueness_xi}

Let us define the vector fields $\xi$ and $\hat{\xi}$ as the solutions to the equations $\xi\iprod e^i= e_{\mathrm{t}}^i$ and ${\hat{\xi}}\iprod h^i=\widehat{e}_{\mathrm{t}}^{\,\,i}$ respectively. These can be written in terms of dual objects as
\begin{equation*}
  \xi=e_{\mathrm{t}i}\left(\frac{\cdot\wedge\frac{1}{2}\varepsilon^{ijk}e_j\wedge e_k}{\mathsf{vol}_e}\right)\,,\quad \hat{\xi}=\frac{1}{1+\alpha^2}\widehat{e}_{\mathrm{t}i}\left(\frac{\cdot\wedge\left(\frac{1}{2}\varepsilon^{ijk}e_j\wedge e_k+\alpha^j e^i\wedge e_j\right)}{\mathsf{vol}_e}\right)\,.
\end{equation*}
Expanding now, we find
\[
\frac{1}{1+\alpha^2}\widehat{e}_{\mathrm{t}i}\left(\frac{1}{2}\varepsilon^{ijk}e_j\wedge e_k+\alpha^j e^i\wedge e_j\right)=\frac{1}{2}\varepsilon^{ijk}e_{\mathrm{t}i}e_j\wedge e_k\,,
\]
where, in the course of the computations, one has to make use of
\begin{align*}
  \frac{1}{2}(e_{\mathrm{t}}\cdot\alpha)\varepsilon_{ijk}\alpha^ie^j\wedge e^k&=\frac{1}{2}e_{\mathrm{t}}^\ell(\varepsilon_{ijk}\alpha_\ell)\alpha^ie^j\wedge e^k=\frac{1}{2}e_{\mathrm{t}}^\ell(\varepsilon_{jk\ell}\alpha_i-\varepsilon_{k\ell i}\alpha_j+\varepsilon_{\ell ij}\alpha_k)\alpha^ie^j\wedge e^k\\
  &=\frac{1}{2}\alpha^2\varepsilon_{ijk}e_{\mathrm{t}}^ie^j\wedge e^k+\varepsilon_{ijk}e_{\mathrm{t}}^i\alpha^j e^k\wedge \mathrm{d}\phi\,.
\end{align*}

\section{\texorpdfstring{Some details on the GNH analysis of the Extended Husain-Kucha\v{r} action}{Some details on the GNH analysis of the Extended Husain-Kucha\v{r} action}}\label{appendix_details_EHK}


In order to guarantee the consistency of the dynamics defined by the vector fields satisfying \eqref{E1}-\eqref{E3} it is necessary to show that they are tangent to the submanifold $\mathsf{M}_1\subset\mathsf{M}_0$ defined by the secondary constraints \eqref{L1}-\eqref{L2}. To this end, we first have to solve \eqref{E1}-\eqref{E3} for $Z_e^i$, $Z_A^i$ and $Z_\phi$. Before that it is convenient to study with some care the secondary contraints of the EHK model.

%
%
\subsection{Constraints}\label{sec_constraints}

Throughout this section we use the objects defined to study the standard HK model, in particular $\mathbb{F}_{ij}$ and $\mathbb{B}_{ij}$. We also expand $\mathrm{d}\phi=\alpha_i e^i$. By doing this we see that the constraints \eqref{L1} and \eqref{L2} are equivalent to
\begin{align}
  & \mathbb{F}_{ij}\alpha^j-\varepsilon_{ijk}\mathbb{F}^{jk}=0\,, \label{LL1}\\
  & \mathbb{B}_{ij}\alpha^j+\varepsilon_{ijk}\mathbb{B}^{jk}=0\,. \label{LL2}
\end{align}
We write now $\mathbb{F}_{ij}$ and $\mathbb{B}_{ij}$ in terms of a symmetric and traceless part, a pure trace and an antisymmetric part according to
\begin{align}
&\mathbb{F}_{ij}\,=\mathbb{S}_{ij}+\frac{1}{3}\delta_{ij}\mathbb{F}\;+\varepsilon_{ijk}\mathbb{A}^k\,,\label{expansion1}\\
&\mathbb{B}_{ij}=\mathbb{U}_{ij}+\frac{1}{3}\delta_{ij}\mathbb{B}+\varepsilon_{ijk}\mathbb{V}^k\,,\label{expansion2}
\end{align}
with $\mathbb{S}_{[ij]}=\mathbb{U}_{[ij]}=0$ and $\mathbb{S}^i_{\phantom{i}i}=\mathbb{U}^i_{\phantom{i}i}=0$. Notice that
\[
\mathbb{F}=\mathbb{F}^i_{\phantom{i}i}\,,\quad \mathbb{B}=\mathbb{B}^i_{\phantom{i}i}\,,\quad \varepsilon_{ijk}\mathbb{F}^{jk}=2\mathbb{A}_i\,,\quad \varepsilon_{ijk}\mathbb{B}^{jk}=2\mathbb{V}_i\,.
\]
Plugging \eqref{expansion1} and \eqref{expansion2} into \eqref{LL1} and \eqref{LL2} we get the following relations
\begin{align}
  (\varepsilon_{ijk}\alpha^k+2\delta_{ij})\mathbb{A}^j &=\mathbb{S}_{ij}\alpha^j\,+\frac{1}{3}\alpha_i\mathbb{F}\,,\label{antisymm1}\\
  (\varepsilon_{ijk}\alpha^k-2\delta_{ij})\mathbb{V}^j &=\mathbb{U}_{ij}\alpha^j+\frac{1}{3}\alpha_i\mathbb{B}\,.\label{antisymm2}
\end{align}
They can be solved for $\mathbb{A}_i$ and $\mathbb{V}_i$ in terms of $(\mathbb{S}_{ij},\mathbb{S})$ and $(\mathbb{U}_{ij},\mathbb{U})$ respectively. The result is
\begin{align}
   \mathbb{A}_i=&{\phantom{-}\,\;}\frac{1}{2(4+\alpha^2)}(4\delta_{ij}+\alpha_i\alpha_j-2\varepsilon_{ijk}\alpha^k)\mathbb{S}^j_{\,\,\ell}\alpha^\ell+\frac{1}{6}\alpha_i \mathbb{F}\,, \label{antisymm01}\\
   \mathbb{V}_i=&-\frac{1}{2(4+\alpha^2)}(4\delta_{ij}+\alpha_i\alpha_j+2\varepsilon_{ijk}\alpha^k)\mathbb{U}^j_{\,\,\ell}\alpha^\ell-\frac{1}{6}\alpha_i \mathbb{B}\,, \label{antisymm02}
\end{align}
where $\alpha^2:=\alpha_i \alpha^i$. Plugging \eqref{antisymm01}-\eqref{antisymm02} into \eqref{expansion1}-\eqref{expansion2} gives $\mathbb{F}_{ij}$ and $\mathbb{B}_{ij}$ in terms of the arbitrary objects $(\mathbb{S}_{ij},\mathbb{F})$ and $(\mathbb{U}_{ij},\mathbb{B})$. Notice, however, that in order to fully solve the secondary constraints one has to find $A_i$ and $e_i$ in terms of $(\mathbb{S}_{ij},\mathbb{F})$ and $(\mathbb{U}_{ij},\mathbb{B})$.

%
%
\subsection{The Hamiltonian vector fields}\label{sec_vector_fields}

A convenient way to solve equations \eqref{E1}-\eqref{E3} for the Hamiltonian vector fields is to put them in the form \eqref{EE1}-\eqref{EE3}. These equations are similar to those of \cite{BarberoG:2023qih} and can be solved by using the methods explained there. Hence, we only give the solution here. As the dynamics takes place in $\mathsf{M}_1$, we use that the secondary constraints are satisfied whenever necessary. By writing $X_A^i=W^i_{\phantom{i}j}e^j$ and $X_e^i=E^i_{\phantom{i}j}e^j$, equations \eqref{EE1} and \eqref{EE2} become
\begin{align}
  &W_{ij} -\delta_{ij}W^k_{\phantom{k}k}-\varepsilon_i^{\phantom{i}k\ell}\alpha_kW_{j\ell}+Z^+_{jk}\mathbb{F}^k_{\phantom{k}i}=0\,, \label{equ_W}\\
  &E_{ij} \,-\delta_{ij\,}E^k_{\phantom{k}k}\,+\varepsilon_i^{\phantom{i}k\ell}\alpha_k\,E_{j\ell}\,+Z^-_{jk}\mathbb{B}^k_{\phantom{k}i}=0\,, \label{equ_E}
\end{align}
where
\[
Z^{\pm}_{ik}:=\varepsilon_{ijk}e_{\mathrm{t}}^j\pm \delta_{ik}Z_\phi\,.
\]
The solutions to \eqref{equ_W} and \eqref{equ_E} can be written as
\begin{align}
  W_{ij} &=\frac{1}{1+\alpha^2}(\delta_i^{\phantom{i}k}+\alpha_i\alpha^k-\varepsilon_i^{\phantom{i}k\ell}\alpha_\ell)M_{kj},\label{solution_W} \\
  E_{ij}\, &=\frac{1}{1+\alpha^2}(\delta_i^{\phantom{i}k}+\alpha_i\alpha^k+\varepsilon_i^{\phantom{i}k\ell}\alpha_\ell)\,N_{kj},\label{solution_E}
\end{align}
where
\begin{align}
M_{ij}&=\phantom{-}e_{\mathrm{t}}^k\mathbb{F}_{ki}\alpha_j-\varepsilon_{jpq}e_{\mathrm{t}}^p\mathbb{F}^q_{\phantom{q}i}-Z_\phi\mathbb{F}_{ji}+\mathbb{F}\left(\frac{1}{2}\delta_{ij}Z_\phi
+\frac{1}{2}\delta_{ij}(e_{\mathrm{t}}\cdot\alpha)-e_{\mathrm{t}i}\alpha_j\right)\,,\label{M}\\
N_{ij}&=-e_{\mathrm{t}}^k\mathbb{B}_{ki}\alpha_j-\varepsilon_{jpq}e_{\mathrm{t}}^p\mathbb{B}^q_{\phantom{q}i}+Z_\phi\mathbb{B}_{ji}\!-\mathbb{B}\left(\frac{1}{2}\delta_{ij}Z_\phi
+\frac{1}{2}\delta_{ij}(e_{\mathrm{t}}\cdot\alpha)-e_{\mathrm{t}i}\alpha_j\right)\,.\label{N}
\end{align}
Here we use the obvious notation $e_{\mathrm{t}}\cdot\alpha:=e_{\mathrm{t}}^i\alpha_i$. For future reference, we note that
\begin{equation}\label{wi}
\varepsilon_{ijk}W^{jk}=e_{\mathrm{t}}^j\mathbb{F}_{ji}-e_{\mathrm{t}i}\mathbb{F}\,.
\end{equation}
In terms of the objects introduced above, equation \eqref{EE3} becomes
\begin{equation}\label{tercera_ec}
\mathbb{B}_{ij}W^{ij}\!+\mathbb{F}_{ij}E^{ij}=0\,.
\end{equation}
If we now plug \eqref{solution_W}-\eqref{solution_E} into \eqref{tercera_ec} and make use of the secondary constraints \eqref{LL1}-\eqref{LL2}, we can check that this equation is identically satisfied. Hence, the $Z_\phi$ are left arbitrary. In this case, the computation is easier than a similar one appearing in the treatment of the self-dual action \cite{BarberoG:2023qih}. The simplest way to perform it is to use \eqref{LL1}-\eqref{LL2} to replace $\mathbb{F}_{ij}\alpha^j$, $\mathbb{B}_{ij}\alpha^j$ by $2\mathbb{A}_i$ and  $-2\mathbb{V}_i$, whenever possible, and then write everything in terms of $\mathbb{S}_{ij}$, $\mathbb{F}$, $\mathbb{A}_i$, $\mathbb{U}_{ij}$, $\mathbb{B}$ and $\mathbb{V}_i$.

%
%
\subsection{Tangency analysis}\label{sec_tangency}

In the GNH framework, the consistency of the Hamiltonian dynamics is tantamount to the tangency  to the submanifold $\mathsf{M}_1$ [defined by the secondary constraints \eqref{L1}-\eqref{L2}] of the Hamiltonian vector fields given by the solutions to \eqref{E1}-\eqref{E3}. The tangency conditions in the present case are

\begin{align}
  & \Big(\varepsilon_{ijk}Z_e^j+\delta_{ik}\mathrm{d}Z_\phi\Big)\wedge F^k+\Big(\varepsilon_{ijk}e^j+\delta_{ik}\mathrm{d}\phi\Big)\wedge DZ_A^k=0\,, \label{tan1}\\
  & D\Big(\varepsilon_{ijk}e^j\wedge Z_e^k-\mathrm{d}Z_\phi\wedge e_i-\mathrm{d}\phi\wedge Z_{ei}\Big)+Z_A^j\wedge\Big(e_i\wedge e_j-\varepsilon_{ijk}\mathrm{d}\phi\wedge e^k\Big)=0\,.\label{tan2}
\end{align}

By following a procedure similar to the one used in the tangency analysis presented in \cite{BarberoG:2023qih}, it is possible to remove the derivatives of $Z_A^i$ and $Z_e^i$ from \eqref{tan1}-\eqref{tan2} and rewrite these tangency conditions on $\mathsf{M}_1$ in the form
\begin{align}
  & \varepsilon_{ijk}X_e^j\wedge F^k+\varepsilon_{ijk}De^j\wedge X_A^k=0\,,\label{tan01}\\
  & (\delta_{ij}e_k-\varepsilon_{ijk}\mathrm{d}\phi)\wedge e^j\wedge X_A^k+(Z_\phi\delta_{ik}+\varepsilon_{ijk}e_{\mathrm{t}}^j)\mathrm{d}\phi\wedge F^k+e_{\mathrm{t}j}e_i\wedge F^j-e_{\mathrm{t}i} F_j\wedge e^j=0\,.\label{tan02}
\end{align}

To check that \eqref{tan02} holds on $\mathsf{M}_1$, we left-wedge multiply by $\mathrm{d}\phi$ the equation \eqref{EE1} that had to be solved in order to find $X_A^i$. This gives
\[
\varepsilon_{ijk}\mathrm{d}\phi\wedge e^j\wedge X_A^k=(\varepsilon_{ijk}e_{\mathrm{t}}^j+\delta_{ik}Z_\phi)\mathrm{d}\phi\wedge F^k\,.
\]
Plugging the previous expression into the tangency condition \eqref{tan02}, we get
\[
e_k\wedge e_i\wedge X_A^k+e_{\mathrm{t}}^j e_i\wedge F_j-e_{\mathrm{t}i} e_j\wedge F^j=0\,.
\]
Now, it is straightforward to see that this condition holds because, as a consequence of \eqref{wi}, we have
\[
e_k\wedge X_A^k=W^{ij}e_i\wedge e_j=\frac{1}{2}(e_{\mathrm{t}}^i\mathbb{F}_{ij}-e_{\mathrm{t}j}\mathbb{F})\varepsilon^{jk\ell} e_k\wedge e_\ell\,.
\]
To check that the tangency condition \eqref{tan01} holds we proceed as we did before with \eqref{tercera_ec}: first write \eqref{tan01} in the equivalent form
\[
\varepsilon_{ijk}\mathbb{B}^{j\ell}W^k_{\phantom{k}\ell}-\varepsilon_{ijk}\mathbb{F}^{j\ell}E^k_{\phantom{k}\ell}=0\,,
\]
use \eqref{LL1}-\eqref{LL2} to replace $\mathbb{F}_{ij}\alpha^j$, $\mathbb{B}_{ij}\alpha^j$ by $2\mathbb{A}_i$ and  $-2\mathbb{V}_i$, and then write everything in terms of $\mathbb{S}_{ij}$, $\mathbb{F}$, $\mathbb{A}_i$, $\mathbb{U}_{ij}$, $\mathbb{B}$ and $\mathbb{V}_i$. The computations are a bit tedious but straightforward.

We end by insisting on the importance of checking tangency conditions. If they had failed to hold, new secondary constraints and/or restrictions on the objects that remained arbitrary at this stage could have appeared. A careful consistency analysis of any new conditions arising in this would have then been necessary.

\end{appendices}

%
%

\end{document}